\documentclass[aps,superscriptaddress,twocolumn,linenumbers, nofootinbib]{revtex4}
\usepackage{lipsum, amsmath, amsthm, amsfonts, physics, adjustbox, bm, graphicx, dsfont, wrapfig, enumerate}
\usepackage{xcolor}
\usepackage{appendix}
\usepackage{comment}
\usepackage{tikz}
\usepackage{thmtools,thm-restate}
\usepackage[hidelinks]{hyperref}
\usepackage{cleveref}

\newcommand{\bsc}{Barcelona Supercomputing Center, Plaça Eusebi G\"uell, 1-3, 08034 Barcelona, Spain}
\newcommand{\ub}{Universitat de Barcelona, 08007 Barcelona, Spain}
\newcommand{\aqa}{{$\langle aQa^L\rangle$} Applied Quantum Algorithms, Universiteit Leiden, Netherlands}
\newcommand{\lorentz}{Lorentz Instituut, Universiteit Leiden, Niels Bohrweg 2, 2333 CA Leiden, Netherlands}
\newcommand{\hri}{Honda Research Institute Europe GmbH, Carl-Legien-Strasse 30, 63073 Offenbach am Main, Germany}
\renewcommand{\vec}[1]{\boldsymbol{#1}}

\newcommand{\Es}[1]{\ensuremath{\mathbb E_{\ket\psi \in S}\left[ #1 \right]}}
\newcommand{\Ex}[1]{\ensuremath{\mathbb E_{x \in \mathcal{X}}\left[ #1 \right]}}
\newcommand{\muts}[2][t]{\ensuremath{\mu_{#1}(#2, S)}}
\newcommand{\mutbars}[2][t]{\ensuremath{\bar\mu_{#1}(#2, S)}}
\newcommand{\mutx}[1][t]{\ensuremath{\mu_{#1}\left(\hat Z_{y}^{(b)}, \mathcal X\right)}}
\newcommand{\sigmax}[1][]
{\ensuremath{\sigma^{#1}\left(\hat Z_{y}^{(b)}, \mathcal X\right)}}
\newcommand{\mutbarx}[1][t]{\ensuremath{\bar\mu_{#1}\left(\hat Z_{y}^{(b)}, \mathcal X\right)}}
\newcommand{\muthaar}[2][t]{\ensuremath{\mu_{#1}(\hat O)}}
\newcommand{\mutbarhaar}[2][t]{\ensuremath{\bar\mu_{#1}(\hat O)}}
\newcommand{\zx}{\ensuremath{z(\vec x)}}
\newcommand{\zxvar}{\ensuremath{z_{\vec\theta}(\vec x)}}
\newcommand{\xbd}{\ensuremath{\mathcal X_{\mathcal B, \mathcal D} }}
\newcommand{\xf}{\ensuremath{\mathcal X_{{\rm F}, \vec\theta}}}
\newcommand{\xru}{\ensuremath{\mathcal X_{{\rm RU}, \vec\theta}}}

\newcommand{\prob}[1]{\ensuremath{\operatorname{Prob}\left( #1 \right)}}

\newcommand{\Zdlp}{\hat{Z}_s}

\newtheorem{definition}{Definition}

\newtheorem{theorem}{Theorem}

\usetikzlibrary{quantikz2}

\begin{document}

\title{The role of data-induced randomness in quantum machine learning classification tasks}
\author{Berta Casas}
\affiliation{\bsc}
\affiliation{\ub}
\author{Xavier Bonet-Monroig}
\affiliation{\aqa}
\affiliation{\lorentz}
\affiliation{\hri}
\author{Adrián Pérez-Salinas}
\affiliation{\aqa}
\affiliation{\lorentz}

\begin{abstract}
Quantum machine learning (QML) has surged as a prominent area of research with the objective to go beyond the capabilities of classical machine learning models.
A critical aspect of any learning task is the process of data embedding, which directly impacts model performance.
Poorly designed data-embedding strategies can significantly impact the success of a learning task.
Despite its importance, rigorous analyses of data-embedding effects are limited, leaving many cases without effective assessment methods.
In this work, we introduce a metric for binary classification tasks, the \textit{class margin}, by merging the concepts of average randomness and classification margin.
This metric analytically connects data-induced randomness with classification accuracy for a given data-embedding map.
We benchmark a range of data-embedding strategies through \textit{class margin}, demonstrating that data-induced randomness imposes a limit on classification performance.
We expect this work to provide a new approach to evaluate QML models by their data-embedding processes, addressing gaps left by existing analytical tools.
\end{abstract}
\maketitle

The promising pace at which theoretical and experimental quantum computing is progressing has motivated scientists world-wide to search for new applications of this technology.
While initial demonstrations of the power of quantum computers will most likely come from specific problems with proven quantum speed-ups, e.g. quantum simulations or Shor's algorithm~\cite{shor1997polynomialtime, feynman1982simulating}, exploring broader applications is of great importance in order to justify the investment.

Quantum Machine Learning (QML) is among the most promising of such unconventional applications.
The field of machine learning encapsulates a class of computational methods that seeks to uncover hidden patterns in data.
Similarly, QML aims at predicting properties of data by using quantum computing pipelines.
There have been promising results demonstrating the power of quantum data~\cite{huang2021power} processed on a quantum computer.
Further results have been able to show that for highly structured data sets, QML can learn properties more efficiently than classical methods~\cite{liu2021rigorous, molteni2024exponential, gyurik2023exponential, wakeham2024inference,gil-fuster2024relation}.
In contrast, when the data does not exhibit any apparent structure, the use of variational approaches have been popularized in order to find potential hidden properties by optimizing controllable parameters with respect to a cost function~\cite{schuld2021supervised, havlicek2019supervised, perez-salinas2020data}.
The heuristic nature of variational classical and quantum machine learning models hinders rigorous complexity-theoretical analysis, yet the experience from classical machine learning has taught us to study the power of heuristic methods.

Thus, it is pivotal to develop new machinery to characterize QML models, for instance using statistical tools.
Perhaps the most notable of such features is the barren plateaus (BPs)~\cite{mcclean2018barren} phenomenon, which refers to the hardness of optimizing a variational model due to exponentially vanishing gradients.
In this context, vanishing gradients are a consequence of parameter-dependent states approximating $t$-designs, i.e., resembling a Haar-random distribution \cite{holmes2022connecting}, which has been dubbed as \textit{expressibility}~\cite{sim2019expressibility, cerezo2021cost, larocca2023theory, larocca2022diagnosing}.

In this work, we propose a method to assess the suitability of the data-embedding map to conduct classification.
To this end, we establish a connection between data-induced randomness and the performance of QML models for binary classification tasks by combining average randomness~\cite{bonet-monroig2024verifying} and the concept of margin in classification tasks \cite{bartlett1996fatshattering}.
We define a new metric, \textit{class margin}, which quantifies classification accuracy.
The concept of margin used in this manuscript is also prominent in generalization learning theory~\cite{vapnik2000nature}, where the so-called fat-shattering dimension arises from it.
However, here we define class margin for characterizing data embeddings.

The main contribution of our manuscript is a set of analytical results showing that the classification accuracy of these models is limited by data-induced randomness, that is, if the quantum states generated by the data-embedding process are approximately drawn from a Haar-random distribution.
We support our analytical findings with three examples: (i) a learning problem with provable quantum advantage that encodes the Discrete Logarithm Problem (DLP)~\cite{liu2021rigorous}, (ii) a tailored task to identify bias in the observable, and (iii) a numerical comparison between variational QML models based on feature maps~\cite{havlicek2019supervised} and data re-uploading \cite{perez-salinas2020data}.

This paper is organized as follows.
In~\Cref{sec.background} we introduce the concept of average randomness and its application to variational models.
Our main result is presented in~\Cref{sec:av_randomness_qml}, where we analytically show how data-induced randomness affects the performance of a classification task.
The three examples are presented in~\Cref{sec:applications}.
In~\Cref{sec.conclusions} we give our conclusions and open questions regarding the effects of data-induced randomness for QML tasks.

\section{Background} \label{sec.background}
\subsection{Binary classification in quantum machine learning}

We focus on QML for binary classification tasks.
In our framework, any QML algorithm comprises two components: (i) an embedding map that transforms data into quantum states, and (ii) an observable used to measure the expectation values of the data-induced quantum states.
The most prominent examples of such framework are variational QML models~\cite{cerezo2021variational, bharti2022noisy}.
In linear models, the data is loaded into quantum states via a fixed embedding map, followed by a parameterized quantum circuit.
This circuit can be interpreted as a variational change of basis.
Data re-uploading models \cite{perez-salinas2020data, schuld2021effect}, on the other hand, can be viewed as employing a tunable embedding map while maintaining a fixed observable.
We refer the reader to \Cref{fig:sphere-data} for a practical visualization of these concepts.

Kernel methods can also be interpreted as linear classification models~\cite{liu2021rigorous}.
In this learning framework, a quantum kernel function that quantifies the similarity between data points, is used as an input for a support vector machine~\cite{cortes1995supportvector} to perform a classification task.
Therefore, there exist a direct equivalence between kernel-based and linear classifications through the representer theorems~\cite{scholkopf2001generalized}.

\subsection{Average randomness}
Throughout this manuscript, we are going to make use of the statistical properties of sets of states $S = \left\{ \ket\psi\right\}$~\cite{bonet-monroig2024verifying}.
The properties of these states are  measured with respect to a given quantum observable $\hat O$, with known spectrum.
In particular, we take advantage of the notion of $\hat{O}$-shadowed statistical moments.
These are the statistical moments of an observable $\hat O$ calculated over the set of states $S$,
\begin{equation}
   \muts{\hat O} = \Es{\bra{\psi} \hat{O}\ket\psi ^t},
\end{equation}
or more conveniently, their centered moments,
\begin{equation}
    \mutbars{\hat O} = \Es{\left(\bra\psi \hat{O}\ket\psi - \mu_1(\hat O, S)\right)^t},
\end{equation}
for $t>1$.
Note that the second moment, $t=2$, $\mutbars[2]{\hat O} \equiv \sigma^2(\hat O, S)$ is exactly the variance of the observable. 
The statistical moments can be used to compare a distribution of quantum states $S$ to the Haar-random distribution.
The difference between these distributions of states, as seen through a given observable $\hat O$, can be quantified through the average anti-randomness
\begin{equation}
    \mathcal{A}_t^{(\hat{O})}(S) = \left\vert\mutbars{\hat O} - \bar\mu_t(\hat{O})\right\vert,
    \label{eq:anti_randomness_def}
\end{equation}
where $\bar\mu_t(\hat O)$ assumes averaging over the Haar-random distribution.
Since the standardized moments are identically $0$ for $t = 1$, first order anti-randomness will become
\begin{equation}
    \mathcal{A}_1^{(\hat{O})}(S) = \left\vert\muts[1]{\hat O} - \mu_1(\hat{O})\right\vert.
\end{equation}
Following this argumentation, if $\mathcal{A}_t^{(\hat{O})}(S) = 0$, then the $t$-th statistical moment of $\bra\psi \hat O \ket\psi$ for $\ket\psi \in S$ is indistinguishable from that of the Haar-random distribution.
\begin{figure}[t!]
    \centering
\includegraphics[width=\linewidth]{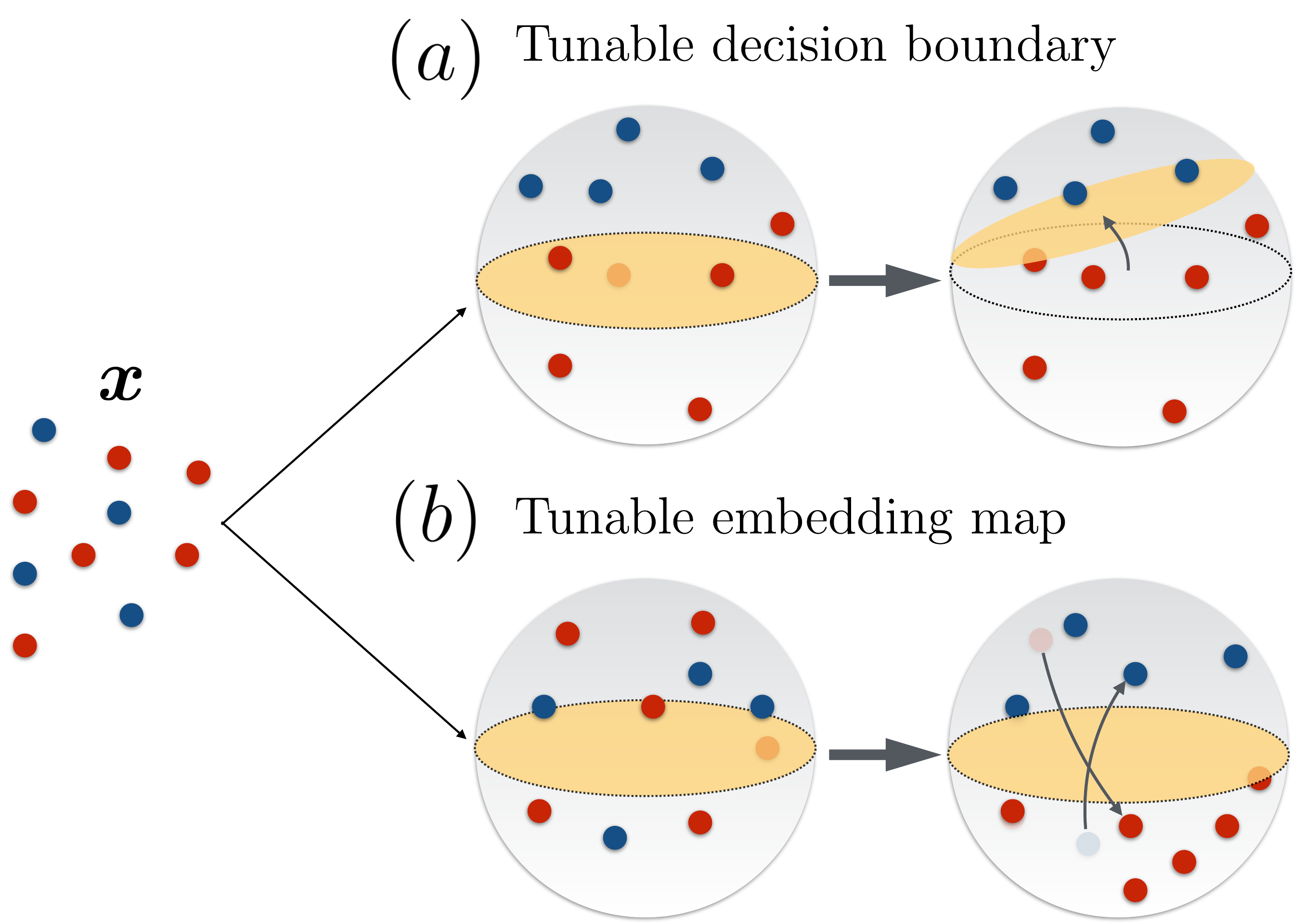}
    \caption{Graphical interpretation of (a) tunable decision boundaries and (b) tunable embedding kernels. In feature-map models, optimization can only provide the optimal observable, and performance is upper bounded by the feature map. Re-uploading models are capable of optimizing the data embedding to perform classification over arbitrary data sets. }
    \label{fig:sphere-data}
\end{figure} 

This notion is related to (spherical) $t$-designs. 
\begin{definition}[(Spherical) $t$-designs~\cite{DELSARTE1976230, ambainis2007quantum}]
    Let $S=\{ \ket\psi\},$ with $\ket\psi \in \mathbb C^N$ be a set of normalized quantum states, and let $\nu(\ket\psi)$ be the Haar measure over states. Then, $S$ is a spherical $t$-design if
    \begin{equation}
        \mathbb E_{\ket\psi \in S}\left[\left(\ket\psi\bra\psi\right)^{\otimes t}\right] = \int_{\ket\phi} d\nu(\phi) \left(\ket\phi\bra\phi\right)^{\otimes t}.
    \end{equation}
\end{definition}

A less restrictive definition of spherical $t$-designs is that of spherical $\hat{O}-$shadowed $t$-designs;
\begin{definition}[$\hat O$-shadowed $t$-design]
   A set of states forms a spherical $\hat{O}$-shadowed $t$-design if $\mathcal{A}_t^{(\hat{O})}(S) = 0$.
\end{definition}
It is important to emphasize that $S$ being $\hat O$-shadowed $t$-design is a necessary but not sufficient condition for $S$ to be a spherical $t$-design.
Furthermore, if $S$ is a $\hat O$-shadowed $t$-design for all positive integers $t$, then $S$ cannot be distinguished from the Haar-random distribution through the observable $\hat O$. 

The statistical moments $\muts{\hat O}$ can be estimated through Monte Carlo sampling over the set $S$. The anti-randomness can also be estimated, since it is possible to analytically compute $\muthaar{\hat O}$~\cite{bonet-monroig2024verifying}. For more details, we refer the reader to Appendix~\ref{appendix:variance_O} and ~\ref{app:centered_moments}. 

\subsection{Effects of randomness in variational quantum algorithms}
Now, we place the notion of average randomness in the context of variational quantum algorithms (VQAs), and show as an example particular to the observed phenomenon of BPs~\cite{mcclean2018barren}.
Extensive research has been devoted to VQAs since its inception~\cite{peruzzo2014variational,bharti2022noisy, cerezo2021variational}, due to its feasibility to be implemented in NISQ~\cite{preskill2018quantum} hardware, and potential to achieve quantum advantage in pre-fault-tolerant quantum computers.
The central piece of any VQA is the so-called parametrized quantum circuit (PQC); a set of quantum operations with classical knobs that can be tuned to navigate the space of accessible quantum states.
More formally,
\begin{equation}
    U(\vec \theta) = \prod_{i = 1}^M V_i U(\theta_i),
\end{equation}
where $V_i,$ and $ U(\theta_i)$ are unitary gates. The navigation of the space of solutions is done by minimization a cost function (i.e. a quantum observable),
\begin{equation}
    \mathcal L(\vec\theta, \hat O) = \bra{\psi(\bm\theta)}\hat O \ket{\psi(\bm\theta)}.
\end{equation}

A PQC defines a family of quantum states $S_\Theta$,
\begin{equation}
    S_\Theta = \{ \psi(\vec\theta) = U(\bm\theta) \ket 0 \vert \vec\theta \sim \Theta \},
\end{equation}
with $U(\vec\theta)$ the actual circuit, and $\Theta$ a distribution of the parameters.

The statistical properties of the optimization landscape induced by $\mathcal L(\vec\theta, \hat O)$ can be inferred from the first two shadowed statistical moments $\mu_t(\hat O, S_\Theta)$:
\begin{multline}
    \operatorname{Var}_{\theta \sim \Theta}\left[\mathcal L(\vec\theta, \hat O)\right] = \sigma^2(\hat O, S_\Theta) = \\  \mu_2 (\hat{O}, S_\Theta) - \mu_1^2 (\hat{O}, S_\Theta).
\end{multline}
Therefore, the notion of randomness is tightly linked to properties of $\mathcal L(\vec\theta, \hat O)$. 
This has been extensively studied in the context of BPs \cite{mcclean2018barren}, manifested as an exponential-in-qubits concentration of the aforementioned quantity, specifically:
\begin{equation}
    \sigma^2(\hat O, S_\Theta) \in e^{-\Omega(n)}. 
\end{equation}
Using the shadowed moments, we can upper-bound the variance of the loss by using the triangular inequality as
\begin{equation}
    \sigma^2(\hat O, S_\Theta) \leq \muthaar[2]{\hat O} + \mathcal{A}_2^{(\hat{O})}(S_\Theta).
\end{equation}
This observation recovers known results from the barren plateau literature~\cite{holmes2022connecting, arrasmith2021effect}, identifying $\mathcal{A}_2^{(\hat{O})}(S_\Theta)$ as the distance to $2$-designs, and arguing that $\muthaar[2]{\hat O} \in e^{-\Omega(n)}$ for $2$-designs. A detailed statement of this property is available in~Appendix \ref{appendix:variance_O}.
Taking advantage of the $\hat{O}-$shadowed $t$-moments, we have a compact and robust framework for analyzing the statistical properties of variational models.

\section{Data-induced randomness in QML classification tasks}\label{sec:av_randomness_qml}
This section contains the main result of our work, that is, the role of data-induced randomness in the performance QML classification tasks.
We analyze how the data embedding can affect the ability to classify quantum states into different classes.
Note that our results are independent of the training process of the models.

For simplicity, we consider the simple yet relevant example of a binary classification task with a quantum circuit.
However, this framework could naturally extend to multi-classification tasks.
In this problem, the data is inserted in the form of $(\vec x, y)$, with $y \in \{0, 1\}$, and $\vec x \in \mathbb R^m$, for an $m$-dimensional feature space.
The $\bm x$-form data is typically introduced into the quantum computer through a feature map, in the form of $U(\vec x) \in \mathcal{SU}(N)$.
The classification task is then reduced to performing a set of measurement on the data-induced states, \(\ket{\psi(\vec x)} = U(\vec x)\ket{0}^{\otimes n}\),
with a task-dependent observable $\hat O$.
As previously mentioned, this framework captures a large number of  QML models, including \textit{feature embedding}~\cite{havlicek2019supervised}, \textit{kernel methods}~\cite{jerbi2023quantum, kubler2021inductive, liu2021rigorous, rebentrost2014quantum, schuld2021supervised}, and the \textit{data re-uploading}~\cite{perez-salinas2020data, schuld2021effect, vidal2020input}.

From the perspective of average randomness, the set of states is given by 
\begin{equation}
    \mathcal X
    = \left\{ \ket{\psi(\vec x)} \right\}_{\vec x},
\end{equation}
where $\mathcal X$ represents the set of states generated as $\bm x$ runs over the dataset (e.g. the training or test set).

We consider $\hat O$ to be a projector, (i.e. its eigenvalues are $\lambda = \{0, 1\}$).
The outcomes are distributed according to the expected value
\begin{equation}\label{eq.observable_o}
    o(\vec x)
 = \bra{\psi(\vec x)} \hat O \ket{\psi(\vec x)}.
\end{equation}
The label of the classification assigned by the model depends on $o (\vec x)$ and a classification threshold $b\in (0,1)$.
Therefore, $y(\vec x) = 0$ if $o(\vec x)<b$ and $y(\vec x) = 1$ if $o(\vec x)\geq b$. 

\begin{figure}[t!]
    \centering
\includegraphics[width=1\linewidth]{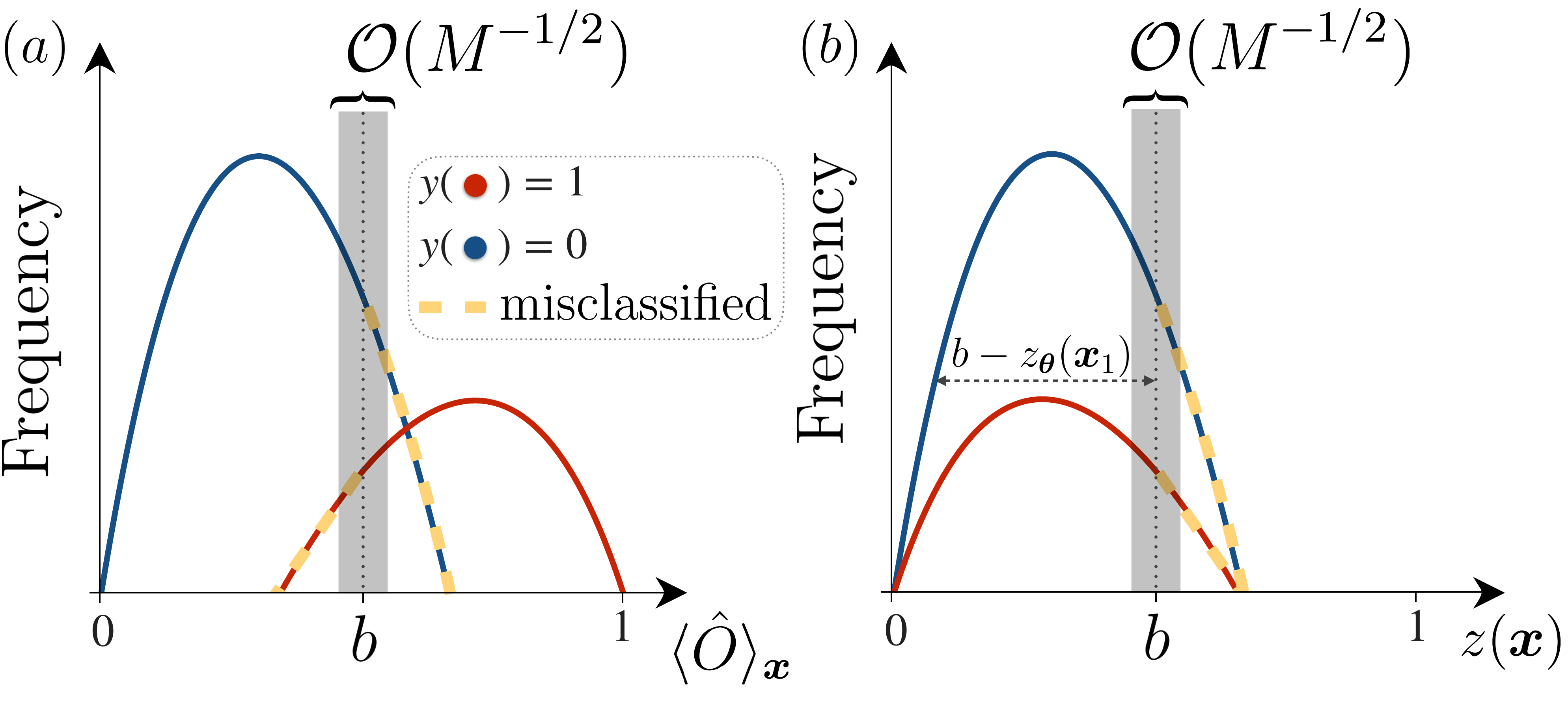}
    \caption{Illustration of the classification criteria and the definition of the class margin $\zx$.
    In both plots, the gray window indicates the region of data points $\mathcal O (M^{-1/2})$ that lie so close to the decision boundary that they cannot be resolved without requiring exponentially many resources as $n$ increases.
    $(a)$ Example of an expected value histogram for a binary classification problem.
    The yellow dashed line represents the misclassified points based on the criteria defined in the text.
    $(b)$ In this plot, data points with $\zx > b$ are  misclassified.
    We also depict the distance to the boundary, $b-\zx$, using a dashed line.}
    \label{fig:z_example}
\end{figure}

To streamline the analysis, we introduce a new variable, \textit{class margin} unifying both classes $\{0 , 1 \}$:
\begin{definition}[Class margin]\label{def.class-margin}
    Let $\mathcal X = \{ \ket{\psi(\vec x)}\}_{\vec x}$ be a set of states generated by a data encoding quantum circuit, where $\vec x$ represents the data.
    Let $\hat O$ be the observable used for classification, measured according to~\Cref{eq.observable_o},
    and let $b$ denote the classification threshold.
    The class margin $\zx$, is then defined as
    \begin{equation}
   \zx = \bra{\psi(\vec x)}   \hat Z^{(b)}_y  \ket{\psi(\vec x)}, 
   \label{eq:class_margin_def}
\end{equation}
where 
\begin{equation}\label{eq:observable_z}
    \hat Z^{(b)}_y = \left\{ \begin{matrix} \hat O & {\rm if} & y(\vec{x}) = 0 
        \\
        f(\hat O, b) & {\rm if} & y(\vec{x}) = 1,
        \end{matrix}\right.,
\end{equation}
where $y(\vec{x})$ is the true label associated to $\vec{x}$ and 
\begin{equation}
    f(\hat O, b) = 
    \left\{ \begin{matrix}
        1 - \frac{(1-b)}{b} \hat O & {\rm if} & 0 \leq o(\vec x) < b  
        \\
        \frac{b}{1-b} (\mathds{I}-\hat O) & {\rm if} & b \leq o(\vec x) \leq 1.
        \end{matrix}\right. 
\label{eq:f(O,b)}
\end{equation}
\end{definition}
A simple example is obtained by selecting $b = 1/2$. In this case, $f(x, 1/2) = 1-x$ if $y(\bm x) = 1$. 

The purpose of class margin is to measure the confidence of a correct classification by quantifying its distance from the decision boundary $b$, rather than an attempt to identify the class itself.
In fact, the class margin is a random variable that depends both on $\vec x$ and $y$, hence, it cannot be used as a predictor. 
However, ~\Cref{def.class-margin} allows for a succinct description of the relationship between data-induced randomness and performance of classification.
For a visual interpretation of class margin we have added~\Cref{fig:z_example}.
Later in the manuscript it will be shown that the statistical moments of the class margin play a crucial role in characterizing the model's performance. 

Here, class margin is purposefully defined to  cover only classification errors due to the data-embedding.
In practice, this translates to the inability to extrapolate the statistical properties of a trained model to test data.
This is contrast to the so-called Generalization Bounds, a family of analytical measures that capture the generalization performance of learning agents~\cite{caro2021encodingdependent, caro2023outofdistribution}.
Yet, a recent work investigated the use of margin-based generalization bounds in QML models \cite{hur2024understandinggeneralizationquantummachine}.

Since $\mathcal{X}$ is a randomly sampled set of states, the class margin $\zx$ is consequently a random variable.
In particular, we are interested in analyzing its statistical moments:
\begin{equation}
    \Ex{\zx^t} \equiv \mutx,
\end{equation}
with the goal of determining the properties of $\zx$ such that accurate classifications are achieved.

To this end, we first consider a fixed value of $\bm x$.
This data point is correctly classified if the corresponding class margin $\zx$ falls below the acceptance threshold $b$.
Additionally, because the classifier returns probabilistic outcomes, it is necessary for $\zx$ to be sufficiently bounded away from $b$ so that a modest number of copies of the state ($M$) will be enough to confidently determine that $\zx \leq b$.
This observation is formalized in the following result. 
\begin{restatable}{lemma}{leclassification}\label{le:classification}
     Consider the class margin $\zx$ for a given data point $\vec{x}$. Suppose the classifier performs $M$ independent measurements of $\zx$ for this data point. Then, for the classifier to correctly classify $\vec{x}$ with probability at least $1 - \delta$, it suffices that
    \begin{equation}
    \zx \leq b - \sqrt{\frac{\log(2/\delta)}{2 M}},
    \end{equation}
    where $b$ is the decision threshold.  
\end{restatable}
The proof is an immediate corollary of Hoeffding's bound for the binomial distribution, and can be found in~Appendix \ref{app:classification}.

The performance of a classifier is measured by the accuracy in the classification of the data points.
We will quantify this performance in terms of the statistical properties of the class margin $\zx$.
A first result from probability theory allows us to bound the classification accuracy.
\begin{restatable}{theorem}{lechebyshev}\label{le:chebyshev}
Consider a quantum classifier defined by the set $\mathcal X_{\bm\theta}$ and the observable $\hat Z^{(b)}_{y(x)}$. The classification is conducted with $M$ copies of the state for each $\bm x$. The probability of failure of classifying a random data point $\vec x$ is given by
\begin{multline}
    \operatorname{Prob}_F \left( \hat Z_y^{(b)}, \mathcal X\right) \leq
    %\prob{\zx \geq b - \sqrt{\frac{\log(2/\delta)}{2 M}}} 
    \\  \frac{\sigmax[2]}{\left(b - \mutx[1] - \sqrt{\frac{\log(2/\delta)}{2 M}}\right)^2}, 
\end{multline}
which includes incorrect classifications and classes not resolved by the measurement uncertainty. 
\end{restatable}
The proof can be found in~Appendix \ref{app:chebyshev}.
The previous result immediately imposes requirements on the $M$ needed to evaluate the classifier that depend on the first and second $\hat{Z}_y^{(b)}-$shadowed moments.
This provides a necessary condition for correct classification:
\begin{restatable}{corollary}{corsamples}\label{th.samples}
    Consider a quantum classifier defined by the set $\mathcal X_{\bm\theta}$ and the observable $\hat Z^{(b)}_{y(x)}$. The classifier correctly classifies a fraction of at least $1 - k$ of the data points with probability at least $1 - \delta$. The number of copies needed for optimal performance is bounded by
    \begin{equation}
        \frac{2M}{\log^2(2/\delta)} \geq \left( b - \mutx[1] - k^{-1} \sigmax \right)^{-2}.
    \end{equation}
\end{restatable}
This equation comes from a direct reformulation of~\Cref{le:chebyshev}.
The interpretation of this result is a tension between the number of copies of the state $M$ required to conduct the classification and the fraction of misclassified points.
Following standard efficiency conventions, $M$ must scale polynomially in $n$.
Under this assumption, an efficient classification is possible only if  
\begin{align}
\label{eq:condition_mu1}
    b - \mutx[1] & \in \Omega\left( {\rm poly}^{-1} (n) \right) \\
    \sigmax[2] & \leq k \left( b - \mutx[1] \right) \nonumber \\ 
    & \qquad \in \mathcal O \left( {\rm poly}^{-1} (n) \right).
    \label{eq:condition_var}
\end{align}
The interpretation of these results is as follows;
in order to classify as many data points as possible, the average of the class margin $\mutx[1]$ must be at least at a distance $1/\operatorname{poly}(n)$ away from the margin $b$, and the variance must be as small as possible.
When the two conditions are met, a correct classification of the majority of the points is guaranteed beyond resolution accuracy.
At this point, we encourage the reader to revisit~\Cref{fig:z_example} to connect the mathematical and the graphical notions of $\zx$.

The conditions in~\cref{eq:condition_mu1} and~\eqref{eq:condition_var} are only satisfied when the set of states produced by the encoding deviates sufficiently from a Haar-random set of states.
This becomes evident in Appendix~\ref{appendix:variance_O}, where the variance of the observable vanishes exponentially in $n$ as the set of states approaches a shadowed $2$-design.
In such cases, the necessary conditions for classification are not met. 

The previous results only take into account the first and second statistical moments of $\zx$.
However, a stronger statement can be formulated under more restrictive conditions over higher order statistical moments;
\begin{restatable}{lemma}{lebernstein}
    \label{le:Bernstein}
     Consider the a quantum classifier defined by the set $\mathcal X_{\bm\theta}$ and the observable $\hat Z^{(b)}_{y}$. The classification is conducted with $M$ copies for each $\bm x$. If the classifier satisfies that
     \begin{equation}
        \mutbarx[t]^{1/t} \leq \sigmax[2] \frac{L}{e} t
        \label{eq:scaling_bernstein}
    \end{equation}
    for a positive constant $L$, then
    \begin{equation}
    \operatorname{Prob}_F \left( \hat Z_y^{(b)}, \mathcal X\right)\leq \exp\left( - \frac{k^2}{2\left(\sigma^2 + L  k \right) }\right), 
\end{equation}
where $k =\left[ b - \sqrt{\frac{\log(2 / \delta)}{2 M}} - \mu_1\right]$.
\end{restatable}
The proof can be found in~Appendix \ref{app:berstein}.

\begin{restatable}{lemma}{lesubgaussian}\label{le:subgaussian}
     Consider a quantum classifier defined by the set $\mathcal X_{\bm\theta}$ and the observable $\hat Z^{(b)}_{y}$. The classification is conducted with $M$ copies for each $\bm x$. If the classifier satisfies that
\begin{equation}
    \mutbarx^{1/t} \leq \frac{L}{\sqrt{2 e}} \sqrt t
    \label{eq:scaling_subgaussian}
\end{equation}
for a positive constant $L$, then
\begin{equation}
     \operatorname{Prob}_F \left( \hat Z_y^{(b)}, \mathcal X\right) \leq \exp\left( - \frac{k^2}{3 L^2}\right), 
\end{equation}
where $k =\left[ b - \sqrt{\frac{\log(2 / \delta)}{2 M}} - \mu_1\right]$.
\end{restatable}
The proof can be found in Appendix~\ref{Appendix:sub-gaussian}. 

The intuition behind these two lemmas is that, if the centered $t$-moments scale sufficiently slow, then we can directly bound the probability of failing in the classification of a data-set.
Note that the difference between~\Cref{eq:scaling_bernstein} and~\eqref{eq:scaling_subgaussian} is that the condition on the $t$-moments of the class margin are bounded by $t$ in the former and by $\sqrt{t}$ in the latter.
These results arise as an immediate consequence of vanishing tails in the distribution of $\zx$.

\section{Detecting data-induced randomness in QML}\label{sec:applications}

In this section, we provide three examples illustrating how the data-induced randomness, captured by the statistical moments of the class margin, reveals critical effects on the performance of classification tasks. 

\subsection{Discrete-Logarithm-Problem-based feature map}\label{sec.dlp}
We consider the task of classifying integer numbers into two classes depending on the solutions to the discrete logarithm problem (DLP).
This problem was the first example of quantum advantage in the domain of QML~\cite{liu2021rigorous}.
The key insight is to embed the classification task into a classical support vector machine algorithm whose kernel is computed with a fault-tolerant quantum computer.
Computing this kernel is possible in BQP (bounded-error quantum polynomial) time~\cite{shor1997polynomialtime}, and it is widely accepted not to be efficiently solvable via classical methods.
Our goal in this section is less ambitious than proving advantages.
We aim to reinterpret these results under the perspective of class margin.
 
We consider a large prime number $p$ and the multiplicative group $ \mathds{Z}_p^* = \{1,2,...,p-1\}$.
We choose a generator $g$ of $\mathds{Z}_p^*$, such that the powers of $g$ span the entire group,
\begin{equation}
   \mathds{Z}_p^* = \left\{g^k (\operatorname{mod } p)\vert k = 1,2,...,p-1\right\}.
\end{equation}
In this setting, every input $x \in \mathds{Z}_p^*$, has a labeling function $y_s(x) \in\{0, 1\}$ coming from a concept class $\mathcal{C} = \{y_s\}_{s\in \mathds{Z}_p^*}$.
Each labeling function is given by  
\begin{equation}
    y_s(x)= \begin{cases} 
    1, & \text { if } \log _g x \in\left[s, s+\frac{p-3}{2}\right] \\ 
    0, & \text { else}.\end{cases}
\end{equation}
The data is encoded into the Hilbert space through the feature map defined by 
\begin{equation}
    x \mapsto\ket{\psi(x)}=\frac{1}{\sqrt{2^k}} \sum_{i=0}^{2^k-1}\left|x \cdot g^i (\operatorname{mod } p)\right\rangle,
\label{eq:DLP_feature_map}
\end{equation}
with $k = n -c\log n$ for some constant $c$.
Hence, the set of states is given by 
\begin{equation}
    \mathcal X_g = \left\{\ket{\psi(x)} \right\}.
    \label{eq:set_states_DLP}
\end{equation}
In this case the only parameters is the generator $g$. 

To interpret this problem from the point-of-view of average randomness, we transform this kernel picture into measurements with respect to projectors~\cite{schuld2021supervised}.
For each $s\in \mathds{Z}_p^*$, or equivalently any concept class $y_s\in \mathcal{C}$, there exists two vectors of the form
\begin{equation}
\begin{split}
    |\psi_s^{(1)}\rangle  &= \frac{1}{\sqrt{(p-1)/2}}\sum_{i = 0}^{(p-3)/2}\ket{g^{s+i} (\operatorname{mod } p)}\\
    |\psi_s^{(0)}\rangle  &= \frac{1}{\sqrt{(p-1)/2}}\sum_{i = (p-3)/2}^{p-1}\ket{g^{s+i} (\operatorname{mod } p)}
\end{split}
\end{equation}
that define a hyperplane that splits the space $\mathcal{X}_g$ in two halves of equal dimension.
These hyperplanes have a large margin property (see~\Cref{appendix:variance_DLP} for details).
We can define the observable $\Zdlp$ as follows
\begin{equation}
    \Zdlp = \frac{\mathds{I}+( \Pi_1-\Pi_0)(-1)^{y_s(x)}}{2},
    \label{eq:Z_s_DLP}
\end{equation}
where $\Pi_0 = |\psi_s^{(0)}\rangle \langle\psi_s^{(0)}| $ and $\Pi_1 = |\psi_s^{(1)}\rangle \langle\psi_s^{(1)}|$.
For simplicity, we are omitting the $\vec x$ dependence on $\hat Z_s$.
Now, the class margin is defined by the expectation value on the set of states given by~\eqref{eq:set_states_DLP}.
Subsequently, the following statistical properties of the class margin are:
\begin{restatable}{lemma}{thmantirandomness}\label{th.antirandomness_dlp}
Consider the set of states given by the feature map in~\Cref{eq:DLP_feature_map} for $x \in \mathds{Z}_p^*$. Let $\Zdlp$ be defined as in~\Cref{eq:Z_s_DLP}. The scaling of $\Zdlp$-shadowed 1- and 2-average anti-randomness of this set of state is given by 
\begin{align}
    \mathcal A_1^{(\Zdlp)}(\mathcal{X}_g) \in & \Theta\left(\frac{1}{{\rm poly}(n)}\right) \\
    \mathcal A_2^{(\Zdlp)}(\mathcal{X}_g) \in & \Theta\left(\frac{1}{{\rm poly}(n)}\right)
\end{align}
\end{restatable}
These results indicate that the set of states $\mathcal X_ g$ are exactly $1/\operatorname{poly}(n)$ bounded away from Haar-random states when measured through $\Zdlp$.
The proof can be found in Appendix~\ref{app:antirandomness_dlp}. 
\begin{restatable}{theorem}{lemDLP}\label{le:DLP1}
    Consider the set of states given by the feature map in~\Cref{eq:DLP_feature_map} for $x \in \mathds{Z}_p^*$. Let $\Zdlp$ be defined as in~\Cref{eq:Z_s_DLP}. Then, the probability of misclassification is bounded by 
    \begin{equation}
    \operatorname{Prob}_F \left(  \Zdlp, \mathcal X_{g}\right) \in \mathcal O\left(\operatorname{poly}^{-1}(n)\right)
    \end{equation}
with a number of copies of the state $M \in \Theta(\operatorname{poly}(n))$.
\end{restatable}
This theorem ensures that we can perform a good classification. The proof of this lemma can be found in~Appendix \ref{appendix:variance_DLP}. 

The ability to classify with high probability is closely related to the set of states $\mathcal{X}_g$ being far from Haar-random states when viewed through $\Zdlp$.
Otherwise, issues on concentration properties arise as we show in~\Cref{appendix:variance_O}.

\subsection{On the role of observable}\label{sec.counterexample}
In this section, we present an ad-hoc classification task that allow us to showcase the use of class margin as a tool to quantify its classification power.
Additionally, we use such a task to study the impact on the choice of the observable on the overall performance.

Our toy model task consist on learning a parameter
$c\in \{0,1\}$ encoded in 
 a data-dependent quantum state through a feature map,
\begin{multline}\label{eq:FM_counterexample}
(c, \vec x) \mapsto  \\
\ket{\psi(c, \vec x)} = \bigotimes_{q = 0}^n \left(\sigma^{(z)}_k\right)^{c} \sum_{q = 0}^n \sqrt{x_k} \ket 0^{\otimes q} \otimes \ket 1^{\otimes n - q},
\end{multline}
with $\bm x\in \mathbb{R}^{n+1}$ the data-vector, and $\sigma^{(z)}_q$ the $Z-$Pauli matrix acting on the $q-$th qubit.
By specifying both $c$ and $\bm x$ distributions the set of quantum states is fixed,
\begin{equation}
    \xbd= \left\{ \ket{\psi(c, \vec x)} | c \sim \mathcal B, \vec x \sim \mathcal D \right\},
    \label{eq:set_of_states_counterexample}
\end{equation}
with $\mathcal B$ a symmetric binomial distribution and $\mathcal{D}$ a Dirichlet distribution~\cite{olkin1964multivariate} defined by $\vec\alpha_k = \frac{1}{2}\binom{n}{k}$.
A random Dirichlet variable $\vec x$ satisfies $\bm x \in [0, 1]^m$, and $\Vert \vec x \Vert_1 = 1$, which ensures the normalization condition of the encoding quantum states in~\Cref{eq:FM_counterexample}.

To learn $c = \{0, 1\}$, we propose two different observables:
\begin{equation}
    \hat{O}_Z = \sum_{q =1}^n \frac{\mathds{I}-\sigma^{(z)}_q}{2n}, 
    \label{eq:observable1}
\end{equation}
which effectively counts the (normalized) number of $1$'s of the quantum state,
and
\begin{equation}
    \hat{O}_X = \frac{\mathds{I}-\sigma^{x}_{\lfloor n/2\rfloor + 1}}{2},
    \label{eq:observable2}
\end{equation}
where $\lfloor n/2\rfloor$ is the largest integer less than or equal to $n/2$, and $\sigma^{(x)}_q$ is the $X$-Pauli matrix acting on the $q-$th qubit.

In what follows, we take advantage of the average anti-randomness~\Cref{eq:anti_randomness_def} metric to assess the inductive bias of both the observables $\hat O_Z$ and $\hat O_X$ in the proposed learning task.
By construction, the set of states $S$ is a $\hat O_Z$-shadowed $t$-design.
The reason behind this is the fact that, when considering Haar-random states $\ket\psi$, and evaluating them with a given observable $\hat O$, it yields an expectation value given by 
\begin{equation}
    \bra{\psi}\hat O \ket{\psi} = \vec\lambda \cdot \bm u.
\end{equation}
Here, $\vec u$ is a random variable drawn from a symmetric Dirichlet distribution, and $\vec\lambda$ are the eigenvalues of $\hat O$~\cite{bonet-monroig2024verifying}.
Aggregation properties of the Dirichlet distribution allow us to simplify the random variable $\vec u$ by reducing it to a lower-dimensional vector, $\vec u^\prime$, sampled from a Dirichlet distribution that accounts for the multiplicities $\vec {m _\lambda}$ of the eigenvalues of $\hat O$, i.e., $\vec u^\prime \sim \mathcal D (\bm {m_\lambda}/2)$. In our case, the observable $\hat O_Z$ has eigenvalues $\lambda$ with multiplicity $m_\lambda$, 
\begin{equation}
    (\lambda, m_\lambda) = \left( \frac{k}{n},\binom{n}{k}\right). 
\end{equation}

The family of states $\xbd$ is designed such that its amplitudes are distributed according to a Dirichlet distribution parameterized to match the multiplicities of $\hat{O}_Z$.
In this way, the statistical moments artificially mimic the Haar-random moments when observed through $\hat{O}_Z$, making the set $\xbd$ a $\hat O_Z$-shadowed $t$-design.
However, the set $\xbd$ is not a $t$-design.
First, this family of quantum states lacks complex phases, which in turn means that covering all elements in the Hilbert space is not possible.
Second, it spans only a restricted region of the Hilbert space, making it incompatible with a $t$-design for any observable other than $\hat{O}_Z$.
As an example, consider observables of the form $\Pi \hat O_Z \Pi^\dagger$, where $\Pi$ is an arbitrary permutation of the elements in the computational basis.
Measuring $\xbd$ with this permuted observable yields expectation values that differ significantly from the original observable with high probability.
The reason lies in the mismatch between the permuted eigenvalues of $\hat{O}_Z$ and the multiplicities encoded in the states.

To study this phenomenon we devise a set of numerical experiments to test the effect of permuting the observable $\hat{O}_Z$ on the average anti-randomness, as defined in~\Cref{eq:anti_randomness_def}.
The results in~\Cref{fig:counterexample}, for $n=8$ qubits, reveal that the anti-randomness of the original observable $\hat O_Z$ is consistent with zero within error bars.
This indicates that all its $t-$moments match those of a Haar-random set of states when measured with respect $\hat O_Z$, as we expected from our ad-hoc data-embedding.
However, introducing permutations to $\hat O_Z$ disrupts the alignment with the structure of the set of states, as we can see in~\Cref{fig:counterexample} with the purple, blue and orange lines.
When applying 1-, 5-, and 15-random permutations to the observable, the value $\mathcal A_t^{ (\hat O_Z)}(S)$ goes from statistically $0$ to higher values.
It is important to notice that the only meaningful comparison is over points with the same $t$, as they are related to the same statistical property.
A first take of this example is the fact that $\xbd$ appears random as seen through $\hat O_Z$, but in reality it is only an artifact of the data-embedding process.

\begin{figure}[t!]                   \centering
\includegraphics[width=1\columnwidth, angle=0]{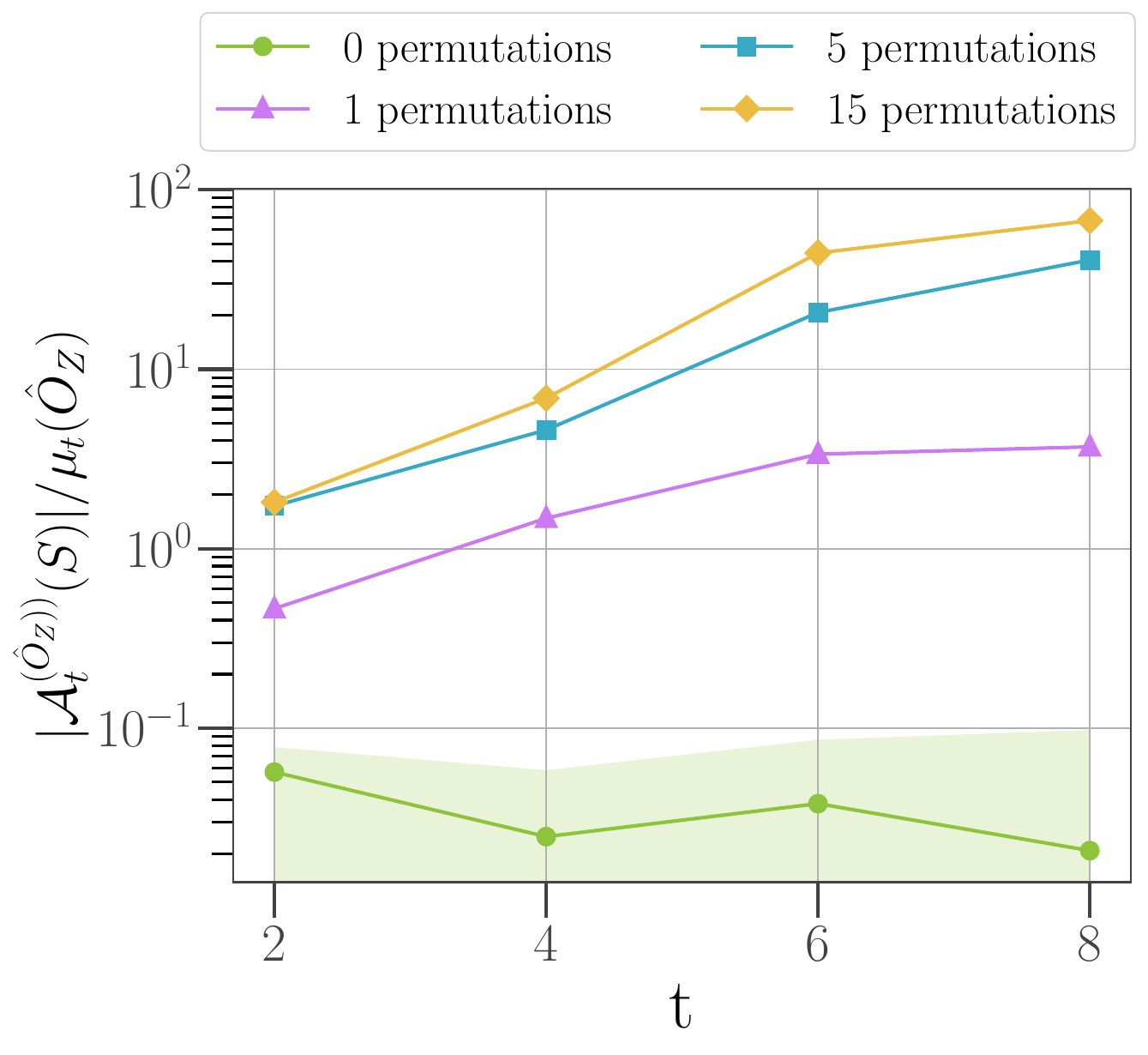}
\caption{Numerical estimation of the anti-randomness $\mathcal{A}_t(S, \hat{O}_Z)$  normalized with respect to $\hat {\mu_t}(\hat{O}_Z)$ and averaged over the set of states $\mathcal{X}_{\mathcal{B}, \mathcal{D}}$ defined in~\Cref{eq:set_of_states_counterexample} for $n=8$. The shaded areas represent the error bars. For the number of necessary samples needed to distinguish if the distribution is a $\hat{O_Z}-$shadowed $t$-design, see Reference~\cite{bonet-monroig2024verifying}. In this case, we tolerate an error of $\epsilon = 0.07$. For the permutation samples, we use $M_\Pi = 2n$. The green line corresponds to the moment computed with the original observable. The other lines correspond to the moments computed with the observable permuted $1, 5$ and $15$ times, respectively.
When, no permutations are applied to the observable, $\mathcal{A}_t(S, \hat{O}_Z)$ is close to zero.
Naively, one would interpret this result as the set of states  $\mathcal{X}_{\mathcal{B}, \mathcal{D}}$ being an $\hat{O}_Z-$shadowed $t$-design.
However, this is far from correct.
As soon as one applies permutations to the observable, the average randomness deviates from zero, and therefore, the family of states is not actually Haar-randomly distributed.}
\label{fig:counterexample}
\end{figure}

Specifically for our learning task, 
the observable $\hat O_Z$ is unable to perform classification, as it views set of states $\xbd$ as completely random.
Thus, $\hat O_Z$ is insensitive to the digit $c$, and the classification shows no inductive bias towards solving the problem.

Additionally, it is possible to show
\begin{align}
    \mu_1(\hat O_Z, \xbd) & = \frac{1}{2} \\
    \sigma^2(\hat O_Z, \xbd) & \in \mathcal O\left(2^{-n}\right).
\end{align}
The observable $\hat O_Z$ will heavily concentrate around its mean, therefore an exponential number of measurements will be required to distinguish between classes, and thus the observable will fail at its job.
See Appendix~\ref{appendix:variance_O} for the analytical formulas of the first moment and the variance averaged over $\hat{O}_Z$ $t$-designs.

Now, we focus on the role of $\hat O_X$ in the classification task on $\xbd$.
In analogy to~\Cref{eq:class_margin_def}, the corresponding random variable that defines the probability of failure is given by 
\begin{equation}
    z(\vec x) = \frac{1}{2} - \sqrt{\vec x_{\lfloor n/2\rfloor } \vec x_{\lceil n/2\rceil } }.
\end{equation}
Here, $\vec x$ is sampled from the Dirichlet distribution previously specified, allowing us to compute the statistical moments analytically.

The specific choice of the set of states and observable allows us to analytically derive the probability to misclassify data points, summarized in the following theorem.
\begin{restatable}{theorem}{thcounterexample}\label{th:counterexample}
    Given the feature map defined in~\Cref{eq:FM_counterexample} and the observable $\hat O_X$, the probability of failure in the classification scales as
    \begin{equation}
    \operatorname{Prob}_F \left( \hat O_X, \mathcal X_{\mathcal{B}, \mathcal{D}}\right)  \in \exp(- \Omega(n))
    \end{equation}
    for $M \in \mathcal O(\operatorname{poly}(n))$. 
\end{restatable}
A proof of this theorem can be found in Appendix \ref{app:counterexample}. This theorem ensures that the observable $\hat{O}_X$ can classify correctly the two families of state. This example illustrates how the $\hat{O}-$shadowed $t$-moments can help us to gain insight into classification bias.
\begin{figure*}[t!]                   \centering
\includegraphics[width=1.8\columnwidth, angle=0, trim={3mm 0mm 0mm 0mm}, clip]{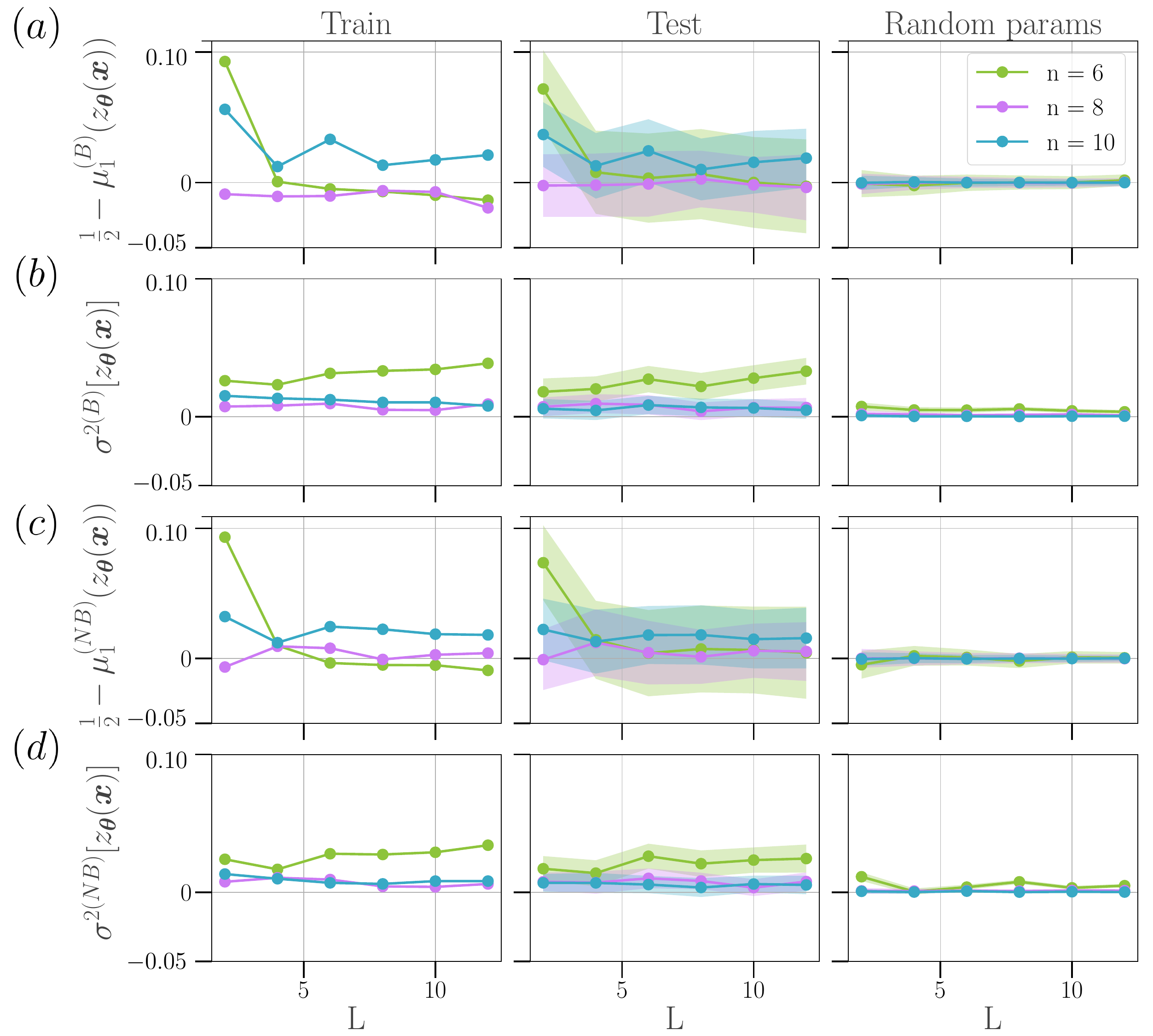}
\caption{Numerical computations for the statistical moments $\mu_1(\zxvar)$, $\sigma^2(\zxvar)$ for feature-map variational QML models, as a function of the number of layers. Results are shown for both the brick and non-brick ansatzes (see Appendix~\ref{Appendix:learning_problem} for details on the circuit). The first row shows the mean and variance over the training set, using optimized parameters obtained via  L-BFGS-B.
The absence of error bars in these figures is due to the fact that we use optimal parameters and a fixed data set, thus the statistical moments can be computed exactly.
The second row displays the mean and variance over the test set sampled from the data distribution. In the third row, parameters $\vec{\theta}$ are sampled randomly from a uniform distribution. The statistical moments are computed via Monte Carlo sampling, and the shaded areas represent the error bars. $(a)$ Mean shifted to $1/2$ and $(b)$ variance of $\zxvar$ for the brick feature map classifier. Mean shifted to $1/2$ and $(b)$ variance of $\zxvar$ for the non-brick feature map classifier.}
\label{fig:RvsL_fm}
\end{figure*}

\subsection{Variational QML models}\label{sec.variational}
Making use of the tools developed so far, we conclude our analysis with a numerical study of the data-encoding induced randomness of two variational QML models.
We remind the reader that we are not concerned about the training aspect of variational QML models, but rather on the potential randomness generated by the data encoding step.
Therefore, we use small models avoiding trainability issues, and evaluate the data-induced randomness after the optimal parameters have been found.
As a proof of concept, we compare linear classification models given by two feature maps~\cite{havlicek2019supervised} and a re-uploading procedure \cite{perez-salinas2020data}, which respectively correspond to~\Cref{fig:sphere-data} (a) and (b).

For the two feature-map models, the data is embedded into the quantum circuits through a fixed data-dependent unitary operation $W(\vec x)$ to yield $\ket{\psi(\vec x)} = W(\vec x) \ket 0$.
One can view the PQC, $U(\vec\theta)$, as a tunable change of basis defining the observable to perform classification (see~\Cref{fig:sphere-data} $(a)$).
The optimization is at most capable of finding an optimal measure for discriminating states from class $0$ and $1$, that is approximating the optimal -- or Helstrom -- measurement~\cite{helstrom1976quantum}.
However, the performance of the classification is upper-bounded by the performance of the feature map itself, which is not trainable. 

In the case of the re-uploading approach, the data is injected by interleaving data-encoding and trainable circuits.
The model can be interpreted as a trainable feature map, where the hyperplane that separates the data is fixed, but the mapping of the data into the quantum feature space is adjustable (see~\Cref{fig:sphere-data} $(b)$). 
Data re-uploading models are universal~\cite{perez-salinas2021one}, hence they are in principle capable of conducting any classification task. 
This property requires the ability to find the optimal parameters, which in practice is difficult due to the concentration of expectation values around the mean~\cite{barthe2023gradients}.
The statistical moments of the class margin provide a direct signal of this phenomenon.

The learning task addressed by both the feature-map and data re-uploading classifiers is a binary classification in two dimensions.
The loss function is
\begin{equation}
      \mathcal{L}(\vec{\theta}) = \sum_{\vec{x}\in X_{\text{train}}} \zxvar,
\end{equation}
where $\zxvar$ is the class margin, and the observable is given by $\hat{Z}_y = \frac{1}{2}(\mathds{I}-(-1)^{y(\vec x)}\sigma^{(z)})$, where $\sigma^{(z)} := \sigma_1^{(z)}\otimes \sigma_2^{(z)}\otimes...\otimes \sigma_n^{(z)}$.
Further details on the learning problem, the quantum circuits, and the optimization are given in~\Cref{Appendix:learning_problem}.

Variational models do not always offer the possibility of a theoretical analysis.
For this reason, we employ numerical experiments to apply class margins to these models and address the validity of our findings. 

\begin{figure*}[t!]                   \centering
\includegraphics[width=0.8\linewidth, angle=0]{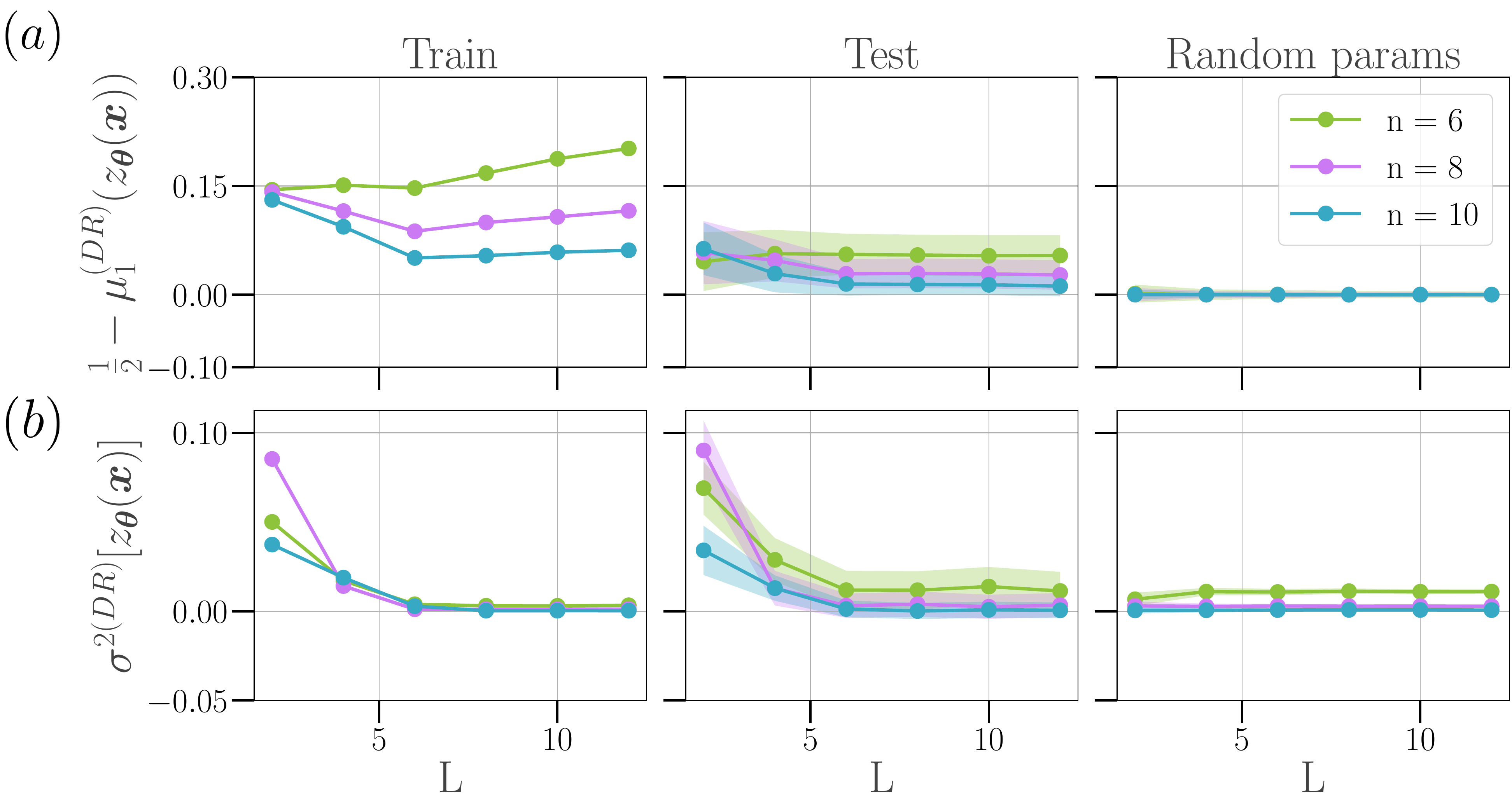}
\caption{Numerical computations for the statistical moments $\mu_1(\zxvar)$, $\sigma^2(\zxvar)$ for data re-uploading model, as a function of the number of layers. The first row shows the mean and variance over the training set, using optimized parameters obtained via L-BFGS-B.
In these plots, error bars are absent because the training set is equispaced.
In this case, Monte Carlo error does not apply, as we are not sampling from a random distribution. 
The second row displays the mean and variance over the test set sampled from the data distribution. In the third row, parameters $\vec{\theta}$ are sampled randomly from a uniform distribution. The statistical moments are computed via Monte Carlo sampling, and the shaded areas represent the error bars. $(a)$ Mean of $\zxvar$ shifted to $1/2$. $(b)$ Variance of $\zxvar$. }
\label{fig:RvsL_DR}
\end{figure*} 
  
\subsubsection*{Numerical results}
We examine the first and second moments of the class margin $\zxvar$ for both feature-map and data re-uploading variational QML models.
To provide a complete description of the model, we choose three different configurations: 1) optimized $\vec\theta$ with $\vec x$ sampled from the training set 2) optimized $\vec\theta$ with $\vec x$ sampled from the test set 3) averaging over randomly distributed $\vec\theta$ values with $\vec x$ sampled from the test set.
1) and 2) give information about how randomness affects model performance, while 3) serve as an analysis of the landscape.

In the feature-map case, we use a \textit{brick}- and \textit{non-brick} data-embedding circuits (see Appendix~\ref{Appendix:learning_problem} for details). 
The results of the numerical experiments are shown in~\Cref{fig:RvsL_fm}.

In the training set (left column), we see that $\mu_1(\zxvar)$ concentrates around $1/2$.
In contrast, $\sigma^2(\zxvar)$ approaches zero for both brick and non-brick feature maps.
In regards to the test set, we observe that both moments trend towards $0$ within error bars as the number of layers and qubit grows.
This can be connected to~\Cref{th.samples}, where we show that having values of $\zxvar$ not sufficiently bounded away from $1/2$ implies an inconsistency in classification.
Therefore, the combination of these two figures indicates a failure in the classification even with the optimal parameters obtained after training.
The results on the test set suggest that the model struggles to generalize effectively.
This aligns with reference~\cite{hur2024understandinggeneralizationquantummachine}, where it is shown that generalization capabilities of QML models are linked with the classification margin.

We validate the results against a randomly sampled set of parameters $\vec \theta$ shown in the third column of~\Cref{fig:RvsL_fm}.
All values fall within $0$, that is, the mean is close to $1/2$ and the variance is $0$.
Using the random parameters as validation, we can state that, as the number of qubits and/or layers increases in the training and test set, the model tends towards being random.
This implies that classification is unfeasible in this scenario, as no observable can effectively discriminate the embedded data. 

Next, we study the data re-uploading model, where the results are depicted in~\Cref{fig:RvsL_DR}.
The first striking trend is the appearance of two regimes in the training set results.
At shallow circuit depths, $L<6$, the values of $\mu_1(\zxvar)$ approach $1/2$ as the number as the layers is increased.
From $L\geq6$ , adding more layers improves the model’s classification performance, which is consistent with the the results in~\cite{perez-salinas2021one}.
However, the improvement in performance is faced with a lack of generalization to the test set.
This can be seen directly from the middle panels, where the centered around $1/2$ mean and the variance tend towards $0$ as $L$ grows.
As a validation step, we use again a data re-uploading process where the parameters are chosen from the uniform distribution.
The results are identical to the previous figure.
As $n$ increases, our numerical results show a slow trend toward the mean.
Even though, the model is universal, it might quickly face trainability issues, consistent with known results in BPs and kernel concentration~\cite{mcclean2018barren, cerezo2021cost, holmes2022connecting, thanasilp2024exponential}.
In fact, all of the above can be traced back to the so-called \textit{curse of dimensionality}, or in other words, the exponential dimension of the Hilbert space.

As a first final remark, our numerical studies confirm our theoretical analysis in \Cref{sec:av_randomness_qml} that a learning problem requiring solutions with uniformly distributed states might be problematic.
A second take from this numerical analysis is that when variational models are executed without strong biases they are inherently random.
This is an indication that both the model architecture and the problem formulation play crucial roles in the randomness and generalization power of the task.

\section{Conclusions}\label{sec.conclusions}
In summary, we have analytically studied the effect of data-induced randomness on the performance of QML models for binary classification tasks.
We have shown that successful classification tasks can only be achieved if the data-induced set of states exhibits limited randomness. 
In other words, the newly introduced metric \textit{class margin} must concentrate below the classification boundary within a distance $\Omega (1/\operatorname{poly}(n))$.
Furthermore, it provides a unified view of the following observations; uniformly exploring the space of quantum states faces the \textit{curse of dimensionality}, and common data embeddings for binary classification make the task impossible due to the concentration properties of the Haar measure.
%and that general data embeddings for binary classifications are an impossible task due to concentration properties of the Haar measure.
The former is linked to trainability limitations in variational quantum algorithms~\cite{mcclean2018barren, holmes2022connecting}.
The latter is aligned with the concentration of kernels~\cite{thanasilp2024exponential} for expressive circuits.

Our general findings are strengthened by applying the framework to three examples.
First, we study a learning problem with provable quantum advantage based on the Discrete Logarithm Problem (DLP)~\cite{liu2021rigorous}.
The success of this algorithm lies in the feature map, which produces a set of states that are significantly distinct from a Haar-random distribution and are believed to be challenging to simulate classically. 
Then, we tackle a tailored task to highlight the effect of the observable on the classification.
In this example, we show that there exist classification tasks for which an observable fails with overwhelming probability, while another observable yields accurate descriptions, hence demonstrating that the choice of the correct observable is crucial.
Finally, we numerically compare variational QML models based on feature maps~\cite{havlicek2019supervised} and data re-uploading \cite{perez-salinas2020data}.
Re-uploading models encode data into quantum states in a flexible manner, thus outperforming feature-maps based models.
However, escaping from random sets of states becomes challenging as the size of the problem increases. 

For the particular case of variational QML models, class margin serves as a diagnostic tool for evaluating the validity of parameter-dependent embeddings.
For each individual configuration of parameters, one can perform Monte Carlo estimations on the relevant statistical moments, and predict the classification power of the model.
Therefore, class margin might be used as a comprehensive performance metric to be optimized during a training phase.

We anticipate that the results of this work will serve as motivation for the community to build new tools and techniques to study the performance of QML tasks. 
In particular, our findings indicate that useful QML methods should avoid data mappings that lead to distributions of states resembling $t$-designs when measured with the observable used for classification.
This insight should encourage the exploration alternative approaches, including applying QML models to highly structured problems, such as the hidden subgroup problems~\cite{wakeham2024inference}.
We expect that progress in this field will contribute to unveil the potential of quantum computing for learning problems. Merging the tools here proposed with quantum advantage analysis will shed light on the applicability of QML.

\acknowledgements
The authors would thank Richard Kueng for pointing them towards Bernstein's inequalities, Matthias C. Caro for his help connecting class margin to generalization bounds, Kristan Temme for insightful perspectives, and Patrick Emonts and Artur Garcia-Saez for feedback on the manuscript.
The authors would like to thank Carlo Beenakker, Jordi Tura, Vedran Dunjko, and Alba Cervera-Lierta for their support on this project.
The authors extend their gratitude to all members of aQa Leiden and BSC's Quantic group for fruitful discussions.
B. C. acknowledges funding from the Spanish Ministry for Digital Transformation and of Civil Service
of the Spanish Government through the QUANTUM ENIA project call - Quantum Spain, EU through
the Recovery, Transformation and Resilience Plan – NextGenerationEU within the framework of the
Digital Spain 2026.
This work was supported by the Dutch National Growth Fund (NGF), as part of the Quantum Delta NL programme, and also funded by the European Union under Grant Agreement 101080142 and the project EQUALITY.
This work was also partially supported by the Dutch Research Council (NWO/OCW), as part of the Quantum Software Consortium programme (project number 024.003.03).
This publication is part of the ``Quantum Inspire - the Dutch Quantum Computer in the Cloud" project (with Project No. NWA.1292.19.194) of the NWA research program ``Research on Routes by Consortia (ORC)", which is funded by the Netherlands Organization for Scientific Research (NWO).

\bibliography{refs.bib}

\begin{thebibliography}{46}
\expandafter\ifx\csname natexlab\endcsname\relax\def\natexlab#1{#1}\fi
\expandafter\ifx\csname bibnamefont\endcsname\relax
  \def\bibnamefont#1{#1}\fi
\expandafter\ifx\csname bibfnamefont\endcsname\relax
  \def\bibfnamefont#1{#1}\fi
\expandafter\ifx\csname citenamefont\endcsname\relax
  \def\citenamefont#1{#1}\fi
\expandafter\ifx\csname url\endcsname\relax
  \def\url#1{\texttt{#1}}\fi
\expandafter\ifx\csname urlprefix\endcsname\relax\def\urlprefix{URL }\fi
\providecommand{\bibinfo}[2]{#2}
\providecommand{\eprint}[2][]{\url{#2}}

\bibitem[{\citenamefont{Shor}(1997)}]{shor1997polynomialtime}
\bibinfo{author}{\bibfnamefont{P.~W.} \bibnamefont{Shor}}, \bibinfo{journal}{SIAM Journal on Computing} \textbf{\bibinfo{volume}{26}}, \bibinfo{pages}{1484} (\bibinfo{year}{1997}), ISSN \bibinfo{issn}{0097-5397}.

\bibitem[{\citenamefont{Feynman}(1982)}]{feynman1982simulating}
\bibinfo{author}{\bibfnamefont{R.~P.} \bibnamefont{Feynman}}, \bibinfo{journal}{International Journal of Theoretical Physics} \textbf{\bibinfo{volume}{21}}, \bibinfo{pages}{467} (\bibinfo{year}{1982}), ISSN \bibinfo{issn}{1572-9575}.

\bibitem[{\citenamefont{Huang et~al.}(2021)\citenamefont{Huang, Broughton, Mohseni, Babbush, Boixo, Neven, and McClean}}]{huang2021power}
\bibinfo{author}{\bibfnamefont{H.-Y.} \bibnamefont{Huang}}, \bibinfo{author}{\bibfnamefont{M.}~\bibnamefont{Broughton}}, \bibinfo{author}{\bibfnamefont{M.}~\bibnamefont{Mohseni}}, \bibinfo{author}{\bibfnamefont{R.}~\bibnamefont{Babbush}}, \bibinfo{author}{\bibfnamefont{S.}~\bibnamefont{Boixo}}, \bibinfo{author}{\bibfnamefont{H.}~\bibnamefont{Neven}}, \bibnamefont{and} \bibinfo{author}{\bibfnamefont{J.~R.} \bibnamefont{McClean}}, \bibinfo{journal}{Nature Communications} \textbf{\bibinfo{volume}{12}}, \bibinfo{pages}{2631} (\bibinfo{year}{2021}), ISSN \bibinfo{issn}{2041-1723}, \eprint{2011.01938}.

\bibitem[{\citenamefont{Liu et~al.}(2021)\citenamefont{Liu, Arunachalam, and Temme}}]{liu2021rigorous}
\bibinfo{author}{\bibfnamefont{Y.}~\bibnamefont{Liu}}, \bibinfo{author}{\bibfnamefont{S.}~\bibnamefont{Arunachalam}}, \bibnamefont{and} \bibinfo{author}{\bibfnamefont{K.}~\bibnamefont{Temme}}, \bibinfo{journal}{Nature Physics} \textbf{\bibinfo{volume}{17}}, \bibinfo{pages}{1013} (\bibinfo{year}{2021}), ISSN \bibinfo{issn}{1745-2481}.

\bibitem[{\citenamefont{Molteni et~al.}(2024)\citenamefont{Molteni, Gyurik, and Dunjko}}]{molteni2024exponential}
\bibinfo{author}{\bibfnamefont{R.}~\bibnamefont{Molteni}}, \bibinfo{author}{\bibfnamefont{C.}~\bibnamefont{Gyurik}}, \bibnamefont{and} \bibinfo{author}{\bibfnamefont{V.}~\bibnamefont{Dunjko}}, \emph{\bibinfo{title}{Exponential quantum advantages in learning quantum observables from classical data}} (\bibinfo{year}{2024}), \eprint{arXiv:2405.02027}.

\bibitem[{\citenamefont{Gyurik and Dunjko}(2023)}]{gyurik2023exponential}
\bibinfo{author}{\bibfnamefont{C.}~\bibnamefont{Gyurik}} \bibnamefont{and} \bibinfo{author}{\bibfnamefont{V.}~\bibnamefont{Dunjko}}, \emph{\bibinfo{title}{Exponential separations between classical and quantum learners}} (\bibinfo{year}{2023}), \eprint{arXiv:2306.16028}.

\bibitem[{\citenamefont{Wakeham and Schuld}(2024)}]{wakeham2024inference}
\bibinfo{author}{\bibfnamefont{D.}~\bibnamefont{Wakeham}} \bibnamefont{and} \bibinfo{author}{\bibfnamefont{M.}~\bibnamefont{Schuld}}, \emph{\bibinfo{title}{Inference, interference and invariance: How the quantum fourier transform can help to learn from data}} (\bibinfo{year}{2024}), \eprint{arXiv:2409.00172}.

\bibitem[{\citenamefont{{Gil-Fuster} et~al.}(2024)\citenamefont{{Gil-Fuster}, Gyurik, {P{\'e}rez-Salinas}, and Dunjko}}]{gil-fuster2024relation}
\bibinfo{author}{\bibfnamefont{E.}~\bibnamefont{{Gil-Fuster}}}, \bibinfo{author}{\bibfnamefont{C.}~\bibnamefont{Gyurik}}, \bibinfo{author}{\bibfnamefont{A.}~\bibnamefont{{P{\'e}rez-Salinas}}}, \bibnamefont{and} \bibinfo{author}{\bibfnamefont{V.}~\bibnamefont{Dunjko}}, \emph{\bibinfo{title}{On the relation between trainability and dequantization of variational quantum learning models}} (\bibinfo{year}{2024}), \eprint{arXiv:2406.07072}.

\bibitem[{\citenamefont{Schuld}(2021)}]{schuld2021supervised}
\bibinfo{author}{\bibfnamefont{M.}~\bibnamefont{Schuld}}, \emph{\bibinfo{title}{Supervised quantum machine learning models are kernel methods}} (\bibinfo{year}{2021}), \eprint{arXiv:2101.11020}.

\bibitem[{\citenamefont{Havl{\'i}{\v c}ek et~al.}(2019)\citenamefont{Havl{\'i}{\v c}ek, C{\'o}rcoles, Temme, Harrow, Kandala, Chow, and Gambetta}}]{havlicek2019supervised}
\bibinfo{author}{\bibfnamefont{V.}~\bibnamefont{Havl{\'i}{\v c}ek}}, \bibinfo{author}{\bibfnamefont{A.~D.} \bibnamefont{C{\'o}rcoles}}, \bibinfo{author}{\bibfnamefont{K.}~\bibnamefont{Temme}}, \bibinfo{author}{\bibfnamefont{A.~W.} \bibnamefont{Harrow}}, \bibinfo{author}{\bibfnamefont{A.}~\bibnamefont{Kandala}}, \bibinfo{author}{\bibfnamefont{J.~M.} \bibnamefont{Chow}}, \bibnamefont{and} \bibinfo{author}{\bibfnamefont{J.~M.} \bibnamefont{Gambetta}}, \bibinfo{journal}{Nature} \textbf{\bibinfo{volume}{567}}, \bibinfo{pages}{209} (\bibinfo{year}{2019}), ISSN \bibinfo{issn}{1476-4687}.

\bibitem[{\citenamefont{{P{\'e}rez-Salinas} et~al.}(2020)\citenamefont{{P{\'e}rez-Salinas}, {Cervera-Lierta}, {Gil-Fuster}, and Latorre}}]{perez-salinas2020data}
\bibinfo{author}{\bibfnamefont{A.}~\bibnamefont{{P{\'e}rez-Salinas}}}, \bibinfo{author}{\bibfnamefont{A.}~\bibnamefont{{Cervera-Lierta}}}, \bibinfo{author}{\bibfnamefont{E.}~\bibnamefont{{Gil-Fuster}}}, \bibnamefont{and} \bibinfo{author}{\bibfnamefont{J.~I.} \bibnamefont{Latorre}}, \bibinfo{journal}{Quantum} \textbf{\bibinfo{volume}{4}}, \bibinfo{pages}{226} (\bibinfo{year}{2020}).

\bibitem[{\citenamefont{McClean et~al.}(2018)\citenamefont{McClean, Boixo, Smelyanskiy, Babbush, and Neven}}]{mcclean2018barren}
\bibinfo{author}{\bibfnamefont{J.~R.} \bibnamefont{McClean}}, \bibinfo{author}{\bibfnamefont{S.}~\bibnamefont{Boixo}}, \bibinfo{author}{\bibfnamefont{V.~N.} \bibnamefont{Smelyanskiy}}, \bibinfo{author}{\bibfnamefont{R.}~\bibnamefont{Babbush}}, \bibnamefont{and} \bibinfo{author}{\bibfnamefont{H.}~\bibnamefont{Neven}}, \bibinfo{journal}{Nature Communications} \textbf{\bibinfo{volume}{9}}, \bibinfo{pages}{4812} (\bibinfo{year}{2018}), ISSN \bibinfo{issn}{2041-1723}.

\bibitem[{\citenamefont{Holmes et~al.}(2022)\citenamefont{Holmes, Sharma, Cerezo, and Coles}}]{holmes2022connecting}
\bibinfo{author}{\bibfnamefont{Z.}~\bibnamefont{Holmes}}, \bibinfo{author}{\bibfnamefont{K.}~\bibnamefont{Sharma}}, \bibinfo{author}{\bibfnamefont{M.}~\bibnamefont{Cerezo}}, \bibnamefont{and} \bibinfo{author}{\bibfnamefont{P.~J.} \bibnamefont{Coles}}, \bibinfo{journal}{PRX Quantum} \textbf{\bibinfo{volume}{3}}, \bibinfo{pages}{010313} (\bibinfo{year}{2022}).

\bibitem[{\citenamefont{Sim et~al.}(2019)\citenamefont{Sim, Johnson, and {Aspuru-Guzik}}}]{sim2019expressibility}
\bibinfo{author}{\bibfnamefont{S.}~\bibnamefont{Sim}}, \bibinfo{author}{\bibfnamefont{P.~D.} \bibnamefont{Johnson}}, \bibnamefont{and} \bibinfo{author}{\bibfnamefont{A.}~\bibnamefont{{Aspuru-Guzik}}}, \bibinfo{journal}{Advanced Quantum Technologies} \textbf{\bibinfo{volume}{2}}, \bibinfo{pages}{1900070} (\bibinfo{year}{2019}), ISSN \bibinfo{issn}{2511-9044}.

\bibitem[{\citenamefont{Cerezo et~al.}(2021{\natexlab{a}})\citenamefont{Cerezo, Sone, Volkoff, Cincio, and Coles}}]{cerezo2021cost}
\bibinfo{author}{\bibfnamefont{M.}~\bibnamefont{Cerezo}}, \bibinfo{author}{\bibfnamefont{A.}~\bibnamefont{Sone}}, \bibinfo{author}{\bibfnamefont{T.}~\bibnamefont{Volkoff}}, \bibinfo{author}{\bibfnamefont{L.}~\bibnamefont{Cincio}}, \bibnamefont{and} \bibinfo{author}{\bibfnamefont{P.~J.} \bibnamefont{Coles}}, \bibinfo{journal}{Nature Communications} \textbf{\bibinfo{volume}{12}}, \bibinfo{pages}{1791} (\bibinfo{year}{2021}{\natexlab{a}}), ISSN \bibinfo{issn}{2041-1723}.

\bibitem[{\citenamefont{Larocca et~al.}(2023)\citenamefont{Larocca, Ju, Garc{\'\i}a-Mart{\'\i}n, Coles, and Cerezo}}]{larocca2023theory}
\bibinfo{author}{\bibfnamefont{M.}~\bibnamefont{Larocca}}, \bibinfo{author}{\bibfnamefont{N.}~\bibnamefont{Ju}}, \bibinfo{author}{\bibfnamefont{D.}~\bibnamefont{Garc{\'\i}a-Mart{\'\i}n}}, \bibinfo{author}{\bibfnamefont{P.~J.} \bibnamefont{Coles}}, \bibnamefont{and} \bibinfo{author}{\bibfnamefont{M.}~\bibnamefont{Cerezo}}, \bibinfo{journal}{Nature Computational Science} \textbf{\bibinfo{volume}{3}}, \bibinfo{pages}{542} (\bibinfo{year}{2023}).

\bibitem[{\citenamefont{Larocca et~al.}(2022)\citenamefont{Larocca, Czarnik, Sharma, Muraleedharan, Coles, and Cerezo}}]{larocca2022diagnosing}
\bibinfo{author}{\bibfnamefont{M.}~\bibnamefont{Larocca}}, \bibinfo{author}{\bibfnamefont{P.}~\bibnamefont{Czarnik}}, \bibinfo{author}{\bibfnamefont{K.}~\bibnamefont{Sharma}}, \bibinfo{author}{\bibfnamefont{G.}~\bibnamefont{Muraleedharan}}, \bibinfo{author}{\bibfnamefont{P.~J.} \bibnamefont{Coles}}, \bibnamefont{and} \bibinfo{author}{\bibfnamefont{M.}~\bibnamefont{Cerezo}}, \bibinfo{journal}{Quantum} \textbf{\bibinfo{volume}{6}}, \bibinfo{pages}{824} (\bibinfo{year}{2022}), ISSN \bibinfo{issn}{2521-327X}, \eprint{2105.14377}.

\bibitem[{\citenamefont{{Bonet-Monroig} et~al.}(2024)\citenamefont{{Bonet-Monroig}, Wang, and {P{\'e}rez-Salinas}}}]{bonet-monroig2024verifying}
\bibinfo{author}{\bibfnamefont{X.}~\bibnamefont{{Bonet-Monroig}}}, \bibinfo{author}{\bibfnamefont{H.}~\bibnamefont{Wang}}, \bibnamefont{and} \bibinfo{author}{\bibfnamefont{A.}~\bibnamefont{{P{\'e}rez-Salinas}}}, \emph{\bibinfo{title}{Verifying randomness in sets of quantum states via observables}} (\bibinfo{year}{2024}), \eprint{arXiv:2404.16211}.

\bibitem[{\citenamefont{Bartlett et~al.}(1996)\citenamefont{Bartlett, Long, and Williamson}}]{bartlett1996fatshattering}
\bibinfo{author}{\bibfnamefont{P.~L.} \bibnamefont{Bartlett}}, \bibinfo{author}{\bibfnamefont{P.~M.} \bibnamefont{Long}}, \bibnamefont{and} \bibinfo{author}{\bibfnamefont{R.~C.} \bibnamefont{Williamson}}, \textbf{\bibinfo{volume}{52}}, \bibinfo{pages}{434} (\bibinfo{year}{1996}), ISSN \bibinfo{issn}{0022-0000}.

\bibitem[{\citenamefont{Vapnik}(2000)}]{vapnik2000nature}
\bibinfo{author}{\bibfnamefont{V.~N.} \bibnamefont{Vapnik}}, \emph{\bibinfo{title}{The {{Nature}} of {{Statistical Learning Theory}}}} (\bibinfo{publisher}{Springer}, \bibinfo{year}{2000}), ISBN \bibinfo{isbn}{978-1-4419-3160-3 978-1-4757-3264-1}.

\bibitem[{\citenamefont{Cerezo et~al.}(2021{\natexlab{b}})\citenamefont{Cerezo, Arrasmith, Babbush, Benjamin, Endo, Fujii, McClean, Mitarai, Yuan, Cincio et~al.}}]{cerezo2021variational}
\bibinfo{author}{\bibfnamefont{M.}~\bibnamefont{Cerezo}}, \bibinfo{author}{\bibfnamefont{A.}~\bibnamefont{Arrasmith}}, \bibinfo{author}{\bibfnamefont{R.}~\bibnamefont{Babbush}}, \bibinfo{author}{\bibfnamefont{S.~C.} \bibnamefont{Benjamin}}, \bibinfo{author}{\bibfnamefont{S.}~\bibnamefont{Endo}}, \bibinfo{author}{\bibfnamefont{K.}~\bibnamefont{Fujii}}, \bibinfo{author}{\bibfnamefont{J.~R.} \bibnamefont{McClean}}, \bibinfo{author}{\bibfnamefont{K.}~\bibnamefont{Mitarai}}, \bibinfo{author}{\bibfnamefont{X.}~\bibnamefont{Yuan}}, \bibinfo{author}{\bibfnamefont{L.}~\bibnamefont{Cincio}}, \bibnamefont{et~al.}, \bibinfo{journal}{Nature Reviews Physics} \textbf{\bibinfo{volume}{3}}, \bibinfo{pages}{625} (\bibinfo{year}{2021}{\natexlab{b}}).

\bibitem[{\citenamefont{Bharti et~al.}(2022)\citenamefont{Bharti, {Cervera-Lierta}, Kyaw, Haug, {Alperin-Lea}, Anand, Degroote, Heimonen, Kottmann, Menke et~al.}}]{bharti2022noisy}
\bibinfo{author}{\bibfnamefont{K.}~\bibnamefont{Bharti}}, \bibinfo{author}{\bibfnamefont{A.}~\bibnamefont{{Cervera-Lierta}}}, \bibinfo{author}{\bibfnamefont{T.~H.} \bibnamefont{Kyaw}}, \bibinfo{author}{\bibfnamefont{T.}~\bibnamefont{Haug}}, \bibinfo{author}{\bibfnamefont{S.}~\bibnamefont{{Alperin-Lea}}}, \bibinfo{author}{\bibfnamefont{A.}~\bibnamefont{Anand}}, \bibinfo{author}{\bibfnamefont{M.}~\bibnamefont{Degroote}}, \bibinfo{author}{\bibfnamefont{H.}~\bibnamefont{Heimonen}}, \bibinfo{author}{\bibfnamefont{J.~S.} \bibnamefont{Kottmann}}, \bibinfo{author}{\bibfnamefont{T.}~\bibnamefont{Menke}}, \bibnamefont{et~al.}, \bibinfo{journal}{Reviews of Modern Physics} \textbf{\bibinfo{volume}{94}}, \bibinfo{pages}{015004} (\bibinfo{year}{2022}).

\bibitem[{\citenamefont{Schuld et~al.}(2021)\citenamefont{Schuld, Sweke, and Meyer}}]{schuld2021effect}
\bibinfo{author}{\bibfnamefont{M.}~\bibnamefont{Schuld}}, \bibinfo{author}{\bibfnamefont{R.}~\bibnamefont{Sweke}}, \bibnamefont{and} \bibinfo{author}{\bibfnamefont{J.~J.} \bibnamefont{Meyer}}, \bibinfo{journal}{Physical Review A} \textbf{\bibinfo{volume}{103}}, \bibinfo{pages}{032430} (\bibinfo{year}{2021}), ISSN \bibinfo{issn}{2469-9926, 2469-9934}, \eprint{2008.08605}.

\bibitem[{\citenamefont{Cortes and Vapnik}(1995)}]{cortes1995supportvector}
\bibinfo{author}{\bibfnamefont{C.}~\bibnamefont{Cortes}} \bibnamefont{and} \bibinfo{author}{\bibfnamefont{V.}~\bibnamefont{Vapnik}}, \bibinfo{journal}{Machine Learning} \textbf{\bibinfo{volume}{20}}, \bibinfo{pages}{273} (\bibinfo{year}{1995}), ISSN \bibinfo{issn}{1573-0565}.

\bibitem[{\citenamefont{Sch{\"o}lkopf et~al.}(2001)\citenamefont{Sch{\"o}lkopf, Herbrich, and Smola}}]{scholkopf2001generalized}
\bibinfo{author}{\bibfnamefont{B.}~\bibnamefont{Sch{\"o}lkopf}}, \bibinfo{author}{\bibfnamefont{R.}~\bibnamefont{Herbrich}}, \bibnamefont{and} \bibinfo{author}{\bibfnamefont{A.~J.} \bibnamefont{Smola}}, in \emph{\bibinfo{booktitle}{International conference on computational learning theory}} (\bibinfo{organization}{Springer}, \bibinfo{year}{2001}), pp. \bibinfo{pages}{416--426}.

\bibitem[{\citenamefont{Delsarte}(1976)}]{DELSARTE1976230}
\bibinfo{author}{\bibfnamefont{P.}~\bibnamefont{Delsarte}}, \bibinfo{journal}{Journal of Combinatorial Theory, Series A} \textbf{\bibinfo{volume}{20}}, \bibinfo{pages}{230} (\bibinfo{year}{1976}), ISSN \bibinfo{issn}{0097-3165}.

\bibitem[{\citenamefont{Ambainis and Emerson}(2007)}]{ambainis2007quantum}
\bibinfo{author}{\bibfnamefont{A.}~\bibnamefont{Ambainis}} \bibnamefont{and} \bibinfo{author}{\bibfnamefont{J.}~\bibnamefont{Emerson}}, \emph{\bibinfo{title}{Quantum t-designs: T-wise independence in the quantum world}} (\bibinfo{year}{2007}), \eprint{quant-ph/0701126}.

\bibitem[{\citenamefont{Peruzzo et~al.}(2014)\citenamefont{Peruzzo, McClean, Shadbolt, Yung, Zhou, Love, {Aspuru-Guzik}, and O'Brien}}]{peruzzo2014variational}
\bibinfo{author}{\bibfnamefont{A.}~\bibnamefont{Peruzzo}}, \bibinfo{author}{\bibfnamefont{J.}~\bibnamefont{McClean}}, \bibinfo{author}{\bibfnamefont{P.}~\bibnamefont{Shadbolt}}, \bibinfo{author}{\bibfnamefont{M.-H.} \bibnamefont{Yung}}, \bibinfo{author}{\bibfnamefont{X.-Q.} \bibnamefont{Zhou}}, \bibinfo{author}{\bibfnamefont{P.~J.} \bibnamefont{Love}}, \bibinfo{author}{\bibfnamefont{A.}~\bibnamefont{{Aspuru-Guzik}}}, \bibnamefont{and} \bibinfo{author}{\bibfnamefont{J.~L.} \bibnamefont{O'Brien}}, \bibinfo{journal}{Nature Communications} \textbf{\bibinfo{volume}{5}}, \bibinfo{pages}{4213} (\bibinfo{year}{2014}), ISSN \bibinfo{issn}{2041-1723}.

\bibitem[{\citenamefont{Preskill}(2018)}]{preskill2018quantum}
\bibinfo{author}{\bibfnamefont{J.}~\bibnamefont{Preskill}}, \bibinfo{journal}{Quantum} \textbf{\bibinfo{volume}{2}}, \bibinfo{pages}{79} (\bibinfo{year}{2018}).

\bibitem[{\citenamefont{Arrasmith et~al.}(2021)\citenamefont{Arrasmith, Cerezo, Czarnik, Cincio, and Coles}}]{arrasmith2021effect}
\bibinfo{author}{\bibfnamefont{A.}~\bibnamefont{Arrasmith}}, \bibinfo{author}{\bibfnamefont{M.}~\bibnamefont{Cerezo}}, \bibinfo{author}{\bibfnamefont{P.}~\bibnamefont{Czarnik}}, \bibinfo{author}{\bibfnamefont{L.}~\bibnamefont{Cincio}}, \bibnamefont{and} \bibinfo{author}{\bibfnamefont{P.~J.} \bibnamefont{Coles}}, \bibinfo{journal}{Quantum} \textbf{\bibinfo{volume}{5}} (\bibinfo{year}{2021}).

\bibitem[{\citenamefont{Jerbi et~al.}(2023)\citenamefont{Jerbi, Fiderer, Poulsen~Nautrup, K{\"u}bler, Briegel, and Dunjko}}]{jerbi2023quantum}
\bibinfo{author}{\bibfnamefont{S.}~\bibnamefont{Jerbi}}, \bibinfo{author}{\bibfnamefont{L.~J.} \bibnamefont{Fiderer}}, \bibinfo{author}{\bibfnamefont{H.}~\bibnamefont{Poulsen~Nautrup}}, \bibinfo{author}{\bibfnamefont{J.~M.} \bibnamefont{K{\"u}bler}}, \bibinfo{author}{\bibfnamefont{H.~J.} \bibnamefont{Briegel}}, \bibnamefont{and} \bibinfo{author}{\bibfnamefont{V.}~\bibnamefont{Dunjko}}, \bibinfo{journal}{Nature Communications} \textbf{\bibinfo{volume}{14}}, \bibinfo{pages}{517} (\bibinfo{year}{2023}), ISSN \bibinfo{issn}{2041-1723}.

\bibitem[{\citenamefont{K{\"u}bler et~al.}(2021)\citenamefont{K{\"u}bler, Buchholz, and Sch{\"o}lkopf}}]{kubler2021inductive}
\bibinfo{author}{\bibfnamefont{J.~M.} \bibnamefont{K{\"u}bler}}, \bibinfo{author}{\bibfnamefont{S.}~\bibnamefont{Buchholz}}, \bibnamefont{and} \bibinfo{author}{\bibfnamefont{B.}~\bibnamefont{Sch{\"o}lkopf}}, \emph{\bibinfo{title}{The {{Inductive Bias}} of {{Quantum Kernels}}}} (\bibinfo{year}{2021}), \eprint{arXiv:2106.03747}.

\bibitem[{\citenamefont{Rebentrost et~al.}(2014)\citenamefont{Rebentrost, Mohseni, and Lloyd}}]{rebentrost2014quantum}
\bibinfo{author}{\bibfnamefont{P.}~\bibnamefont{Rebentrost}}, \bibinfo{author}{\bibfnamefont{M.}~\bibnamefont{Mohseni}}, \bibnamefont{and} \bibinfo{author}{\bibfnamefont{S.}~\bibnamefont{Lloyd}}, \bibinfo{journal}{Physical Review Letters} \textbf{\bibinfo{volume}{113}}, \bibinfo{pages}{130503} (\bibinfo{year}{2014}).

\bibitem[{\citenamefont{Vidal and Theis}(2020)}]{vidal2020input}
\bibinfo{author}{\bibfnamefont{J.~G.} \bibnamefont{Vidal}} \bibnamefont{and} \bibinfo{author}{\bibfnamefont{D.~O.} \bibnamefont{Theis}}, \emph{\bibinfo{title}{Input {{Redundancy}} for {{Parameterized Quantum Circuits}}}} (\bibinfo{year}{2020}), \eprint{arXiv:1901.11434}.

\bibitem[{\citenamefont{Caro et~al.}(2021)\citenamefont{Caro, {Gil-Fuster}, Meyer, Eisert, and Sweke}}]{caro2021encodingdependent}
\bibinfo{author}{\bibfnamefont{M.~C.} \bibnamefont{Caro}}, \bibinfo{author}{\bibfnamefont{E.}~\bibnamefont{{Gil-Fuster}}}, \bibinfo{author}{\bibfnamefont{J.~J.} \bibnamefont{Meyer}}, \bibinfo{author}{\bibfnamefont{J.}~\bibnamefont{Eisert}}, \bibnamefont{and} \bibinfo{author}{\bibfnamefont{R.}~\bibnamefont{Sweke}}, \bibinfo{journal}{Quantum} \textbf{\bibinfo{volume}{5}}, \bibinfo{pages}{582} (\bibinfo{year}{2021}).

\bibitem[{\citenamefont{Caro et~al.}(2023)\citenamefont{Caro, Huang, Ezzell, Gibbs, Sornborger, Cincio, Coles, and Holmes}}]{caro2023outofdistribution}
\bibinfo{author}{\bibfnamefont{M.~C.} \bibnamefont{Caro}}, \bibinfo{author}{\bibfnamefont{H.-Y.} \bibnamefont{Huang}}, \bibinfo{author}{\bibfnamefont{N.}~\bibnamefont{Ezzell}}, \bibinfo{author}{\bibfnamefont{J.}~\bibnamefont{Gibbs}}, \bibinfo{author}{\bibfnamefont{A.~T.} \bibnamefont{Sornborger}}, \bibinfo{author}{\bibfnamefont{L.}~\bibnamefont{Cincio}}, \bibinfo{author}{\bibfnamefont{P.~J.} \bibnamefont{Coles}}, \bibnamefont{and} \bibinfo{author}{\bibfnamefont{Z.}~\bibnamefont{Holmes}}, \bibinfo{journal}{Nature Communications} \textbf{\bibinfo{volume}{14}}, \bibinfo{pages}{3751} (\bibinfo{year}{2023}), ISSN \bibinfo{issn}{2041-1723}.

\bibitem[{\citenamefont{Hur and Park}(2024)}]{hur2024understandinggeneralizationquantummachine}
\bibinfo{author}{\bibfnamefont{T.}~\bibnamefont{Hur}} \bibnamefont{and} \bibinfo{author}{\bibfnamefont{D.~K.} \bibnamefont{Park}}, \emph{\bibinfo{title}{Understanding generalization in quantum machine learning with margins}} (\bibinfo{year}{2024}), \eprint{arXiv: 2411.06919}.

\bibitem[{\citenamefont{Olkin and Rubin}(1964)}]{olkin1964multivariate}
\bibinfo{author}{\bibfnamefont{I.}~\bibnamefont{Olkin}} \bibnamefont{and} \bibinfo{author}{\bibfnamefont{H.}~\bibnamefont{Rubin}}, \bibinfo{journal}{The Annals of Mathematical Statistics} \textbf{\bibinfo{volume}{35}}, \bibinfo{pages}{261 } (\bibinfo{year}{1964}).

\bibitem[{\citenamefont{Helstrom}(1976)}]{helstrom1976quantum}
\bibinfo{author}{\bibfnamefont{C.~W.} \bibnamefont{Helstrom}}, \emph{\bibinfo{title}{Quantum Detection and Estimation Theory}}, no. \bibinfo{number}{v. 123} in \bibinfo{series}{Mathematics in Science and Engineering} (\bibinfo{publisher}{Academic Press}, \bibinfo{address}{New York}, \bibinfo{year}{1976}), ISBN \bibinfo{isbn}{978-0-12-340050-5}.

\bibitem[{\citenamefont{{P{\'e}rez-Salinas} et~al.}(2021)\citenamefont{{P{\'e}rez-Salinas}, {L{\'o}pez-N{\'u}{\~n}ez}, {Garc{\'i}a-S{\'a}ez}, {Forn-D{\'i}az}, and Latorre}}]{perez-salinas2021one}
\bibinfo{author}{\bibfnamefont{A.}~\bibnamefont{{P{\'e}rez-Salinas}}}, \bibinfo{author}{\bibfnamefont{D.}~\bibnamefont{{L{\'o}pez-N{\'u}{\~n}ez}}}, \bibinfo{author}{\bibfnamefont{A.}~\bibnamefont{{Garc{\'i}a-S{\'a}ez}}}, \bibinfo{author}{\bibfnamefont{P.}~\bibnamefont{{Forn-D{\'i}az}}}, \bibnamefont{and} \bibinfo{author}{\bibfnamefont{J.~I.} \bibnamefont{Latorre}}, \bibinfo{journal}{Physical Review A} \textbf{\bibinfo{volume}{104}}, \bibinfo{pages}{012405} (\bibinfo{year}{2021}), ISSN \bibinfo{issn}{2469-9926, 2469-9934}.

\bibitem[{\citenamefont{Barthe and {P{\'e}rez-Salinas}}(2023)}]{barthe2023gradients}
\bibinfo{author}{\bibfnamefont{A.}~\bibnamefont{Barthe}} \bibnamefont{and} \bibinfo{author}{\bibfnamefont{A.}~\bibnamefont{{P{\'e}rez-Salinas}}}, \emph{\bibinfo{title}{Gradients and frequency profiles of quantum re-uploading models}} (\bibinfo{year}{2023}), \eprint{arXiv:2311.10822}.

\bibitem[{\citenamefont{Thanasilp et~al.}(2024)\citenamefont{Thanasilp, Wang, Cerezo, and Holmes}}]{thanasilp2024exponential}
\bibinfo{author}{\bibfnamefont{S.}~\bibnamefont{Thanasilp}}, \bibinfo{author}{\bibfnamefont{S.}~\bibnamefont{Wang}}, \bibinfo{author}{\bibfnamefont{M.}~\bibnamefont{Cerezo}}, \bibnamefont{and} \bibinfo{author}{\bibfnamefont{Z.}~\bibnamefont{Holmes}}, \bibinfo{journal}{Nature Communications} \textbf{\bibinfo{volume}{15}}, \bibinfo{pages}{5200} (\bibinfo{year}{2024}).

\bibitem[{\citenamefont{S.N.Bernstein}(1924)}]{Bernstein}
\bibinfo{author}{\bibnamefont{S.N.Bernstein}}, \bibinfo{journal}{Ann. Sci. Inst. Sav. Ukraine} \textbf{\bibinfo{volume}{4}} (\bibinfo{year}{1924}).

\bibitem[{\citenamefont{Bercu et~al.}(2015)\citenamefont{Bercu, Delyon, and Rio}}]{bercu2015concentration}
\bibinfo{author}{\bibfnamefont{B.}~\bibnamefont{Bercu}}, \bibinfo{author}{\bibfnamefont{B.}~\bibnamefont{Delyon}}, \bibnamefont{and} \bibinfo{author}{\bibfnamefont{E.}~\bibnamefont{Rio}}, \emph{\bibinfo{title}{Concentration {{Inequalities}} for {{Sums}} and {{Martingales}}}}, {{SpringerBriefs}} in {{Mathematics}} (\bibinfo{publisher}{Springer International Publishing}, \bibinfo{year}{2015}), ISBN \bibinfo{isbn}{978-3-319-22099-4}.

\bibitem[{\citenamefont{Wendel}(1948)}]{Wendel}
\bibinfo{author}{\bibfnamefont{J.~G.} \bibnamefont{Wendel}}, \bibinfo{journal}{The American Mathematical Monthly} \textbf{\bibinfo{volume}{55}}, \bibinfo{pages}{563} (\bibinfo{year}{1948}), ISSN \bibinfo{issn}{00029890, 19300972}.

\bibitem[{\citenamefont{Gautschi}(1959)}]{Gautschi}
\bibinfo{author}{\bibfnamefont{W.}~\bibnamefont{Gautschi}}, \bibinfo{journal}{Journal of Mathematics and Physics} \textbf{\bibinfo{volume}{38}}, \bibinfo{pages}{77} (\bibinfo{year}{1959}).

\end{thebibliography}

\onecolumngrid
\newpage

\appendix

\section{Analytic expression of the variance for Haar random-states }
\label{appendix:variance_O}
In this section, we derive an analytic expression for the variance of a given observable $\hat{O}$ when the family of states $S = \{|\psi({\vec{x}})\rangle\}$ forms, at least, an $\hat{O}-$shadowed $2-$design. Recall that a set forming a $2-$design is an $\hat{O}-$shadowed $2-$design for all $\hat{O}$, but an $\hat{O}-$shadowed $2-$design is not necessarily a $2-$design. The variance of $\langle \psi ({\vec{x}})|\hat{O}|\psi({\vec{x}})\rangle$ is given by 
\begin{equation}
    \sigma^2(\hat O, S) = \mu_2(\hat{O}, S) - \left( \mu_1(\hat{O}, S) \right)^2, 
    \label{eq:Var_O}
\end{equation}
We can compute the first moment as 
\begin{equation}
    \mu_1(\hat{O}, S) = \mathds{E}_S[\langle \psi ({\vec{x}})|\hat{O}|\psi({\vec{x}})\rangle] =  \mathds{E}_S[\vec{\lambda}\cdot \vec{u}] = \sum_{i = 1}^G \lambda_i \mathds{E}_S[x_i] = \frac{1}{\alpha_0} \sum_{i = 1}^G \lambda_i \alpha_i ,
    \label{eq:first_moment_Haar}
\end{equation}
where in the second equality we have used the results in Ref~\cite{bonet-monroig2024verifying}, with $\vec{\lambda} = (\lambda_1, \lambda_2, ...\lambda_G)$ being the vector of the $G$ different eigenvalues of $\hat{O}$ and $\vec{u}$ is a random variable sampled according to a Dirichlet distribution with parameter $\vec{\alpha} = \frac{\vec{m}}{2}$, being $\vec{m} = (m_1, m_2, ...m_G)$ the vector of the multiplicities associated with each eigenvalue. In the last equality, we have used that, if $\vec{u}\sim \text{Dir}(\vec{\alpha})$, then $E[u_i] = \alpha_i/\alpha_0$, being $\alpha_0 = \sum_{i = 1}^G \alpha_i$. Now, for the second moment: 
\begin{align}
    \mu_2(\hat{O}, S) = &\mathds{E}_S[\langle \psi ({\vec{x}})|\hat{O}|\psi({\vec{x}})\rangle^2] =  \mathds{E}_S\left[ (\vec{\lambda}\cdot \vec{u})^2 \right]= \mathds{E}_S\left[ \sum_{i = 1}^G \lambda_i u_i \right]^2 = \mathds{E}_S\left[ \sum_{i,j = 1}^G \lambda_i \lambda_j u_i u_j\right] = \\  &\sum_{\substack{i,j=1 \\ i\neq j }}^G \lambda_i \lambda_j \mathds{E}_S[u_i u_j] + \sum_{i}^G \lambda_i^2 \mathds{E}_S[u_i ^2] = \sum_{\substack{i,j=1 \\ i\neq j }}^G \lambda_i \lambda_j\frac{\alpha_i \alpha_j }{\alpha_0 (\alpha_0 +1)} + \sum_{i = 1}^G \frac{\lambda_i^2 \alpha_i (\alpha_i +1)}{\alpha_0 (\alpha_0+1)},
\end{align}
where we have used that $\mathds{E}[u_i u_j] = \frac{\alpha_i \alpha_j}{\alpha_0 (\alpha_0 +1)}$ for $u_i \neq u_j$ and $\mathds{E}[u_i^k] = \frac{\alpha_i (\alpha_i+1)...(\alpha_i+k-1)}{\alpha_0 (\alpha_0+1)...(\alpha_0+k-1)}$. Putting all together in Eq.~\eqref{eq:Var_O}, we end up having 
\begin{align}
\label{eq:variance_Haar_states}
   \sigma^2(\hat O, S) = \sum_{\substack{i,j=1 \\ i\neq j }}^G \lambda_i \lambda_j\frac{\alpha_i \alpha_j }{\alpha_0 (\alpha_0 +1)} + \sum_{i = 1}^G \frac{\lambda_i^2 \alpha_i (\alpha_i +1)}{\alpha_0 (\alpha_0+1)} - \sum_{i,j = 1}^G \lambda_i \lambda_j \frac{\alpha_i \alpha_j}{\alpha_0^2 } = \\\sum_{\substack{i,j=1 \\ i\neq j }}^G \lambda_i \lambda_j\frac{\alpha_i \alpha_j }{\alpha_0} \left( \frac{1}{\alpha_0 +1} - \frac{1}{\alpha_0} \right) + \sum_{i = 1}^G \frac{\lambda_i^2 \alpha_i }{\alpha_0}\left( \frac{\alpha_i +1}{\alpha_0+1}- \frac{\alpha_i}{\alpha_0}\right) = \\ \sum_{i = 1}^G \frac{\lambda_i^2 \alpha_i }{\alpha_0 ( \alpha_0+1)} - \sum_{i,j=1}^G \frac{\lambda_i \lambda_j \alpha_i \alpha_j}{\alpha_0^2 (\alpha_0 +1)} .
\end{align}
In particular, we derive the scaling of the variance when the observable $\hat{O}$ is a projector:
\begin{align}
    \sigma^2(\hat O, S) =&\sum_{i = 1}^G \frac{\lambda_i^2 m_i }{2^{n} ( 2^{n-1}+1)} - \sum_{i,j=1}^G \frac{\lambda_i \lambda_j m_i m_j}{2^{2n} (2^{n-1} +1)} =\\&\frac{1}{2^n (2^{n-1}+1)}\left(\sum_{i = 1}^G {\lambda_i^2 m_i }- \frac{1}{2^n}\sum_{i,j=1}^G  \lambda_i \lambda_j m_i m_j\right) \leq \\&\frac{1}{2^n (2^{n-1}+1)}\sum_{i = 1}^G {\lambda_i^2 m_i } \leq \frac{2^n}{2^n (2^{n-1}+1)}\in \mathcal{O}(2^{-n}).
    \label{eq:variance_haar_random_1}
\end{align}
Where we have used that $\alpha_i = m_i/2$ in $\sum_{i = 1}^G \alpha_i = 2^{n-1}$, and $\lambda_i \in \{0,1\}$. Therefore, the variance of a projector averaged over an $\hat{O}-$-shadowed 2-design family of states $S=\{|\psi\rangle\}$ is bounded by 
    \begin{equation}
       \sigma^2(\hat O, S) \in  e^{-\Omega(n)}.
    \end{equation}
This aligns with the results obtained in reference~\cite{thanasilp2024exponential} for expressive kernels.

\section{Analytical expression for the centered $\hat{O}$-shadowed $t$-moments}
\label{app:centered_moments}
In this section, we derive the analytic expression for the centered $\hat{O}$-shadowed $t$-moments for arbitrary $t$. They are defined as follows: 
\begin{equation}
    \mutbars{\hat O} = \Es{\left(\bra\psi \hat{O}\ket\psi - \mu_1(\hat O, S)\right)^t}.
\end{equation}
Now, let's derive an analytical formula for computing them: 
\begin{align}
   \mutbars{\hat O} = \Es{\sum_{k = 0}^t \binom{t}{k} \bra\psi \hat{O}\ket\psi ^k (-1)^{t-k}\mu_1(\hat O, S)^{t-k}} = \sum_{k = 0}^t\binom{t}{k}(-1)^{t-k}\mu_1^{t-k}(\hat O, S)\muts[k]{\hat O},
\end{align}
Notice that we have an analytical expression for $\muts{\hat O}$~\cite{bonet-monroig2024verifying}. Putting all together, we obtain 
\begin{align}
    \mutbars{\hat O} = \sum_{k = 0}^t\binom{t}{k}(-1)^{t-k}\mu_1(\hat O, S)^{t-k}\sum_{\substack{\boldsymbol{l} \in \mathbb{N}^G \\\|\boldsymbol{l}\|_1=k}}\binom{k}{\boldsymbol{l}}\left(\prod_{i=1}^G \lambda_i^{l_i}\right) \frac{\Gamma\left(\alpha_0\right)}{\Gamma\left(\alpha_0+k\right)} \prod_{i=1}^N \frac{\Gamma\left(\alpha_i+l_i\right)}{\Gamma\left(\alpha_i\right)}
\end{align}
where the sum of the index $\vec l$ runs over all possible non-negative integers $l_1, l_2,...l_G$ with $k = l_1+l_2+...l_G$. Recall that $\alpha_i = m_i/2$ being $m_i$ the multiplicity of the eigenvalue $\lambda_i$ of the observable $\hat O$ and $\alpha_0 = \sum_{i = 1}^G \alpha_i = 2^{n-1}$. We have also used that the first $\hat{O}-$shadowed moment has a simple expression given by $\mu_1(\hat O, S) = \sum_{i = 1}^G \lambda_i \alpha_i /\alpha_0$ (see Appendix~\ref{appendix:variance_O}).

\section{Proof of~\Cref{le:classification}}
\label{app:classification}
\leclassification*
In our classification model, the output of the quantum computation is retrieved through a measurement of the observable $\hat O$. We assume that $\hat O$ is a projector in our model, hence the outcomes are $\{0, 1\}$. The probability of classifying a point in the incorrect class $y'(\vec x)$ is given by the value $\zx$ and its comparisson with the threshold $b$. Hence, the probability of classifying $\vec x$ correctly is given by a binomial distribution with average $1 - \zx$. 

Hoeffding's inequality for binomial distribution implies an exponential-in-samples accuracy in the estimation of $\zx$. Consider $k(\vec x)$ as the number of outcomes of the incorrect class, i.e., the number of times we measure $o(\bm x) = 0$ when $y(\bm x) = 1$ (and viceversa) with a single shoot. Then,  
\begin{equation}
    \prob{\left\vert \frac{k(\vec x)}{M} - \zx\right\vert \geq \epsilon} \leq 2 \exp\left( -2 \epsilon^2 M\right).
\end{equation}

We are interested in determining whether $\zx < 1/2$, in other words, in determining if the classification of $\vec x$ is correct. Our classification will be correct (with certain probability) if our estimation $k(\vec x) / M \leq \frac{1}{2} - \epsilon$, with $\epsilon$ being the error in the estimation depending on the number of samples.  
Considering a confidence level  $\delta$, we can state that, with probability $1 - \delta$,
\begin{equation}
    \left\vert  k(\vec x)/M - \zx \right\vert \leq \sqrt{\frac{\log(2/\delta)}{2 M}},
\end{equation}
and we recover the well-known result of the scaling of the error $\epsilon \in \mathcal O \left(\frac{1}{\sqrt{M}}\right)$. 

Therefore, an estimation of $\zx$ with $M$ measurements allows for a distinction of three categories. With probability at least $1 - \delta$
\begin{align}
    & {\rm if \quad }\frac{k(\vec x)}{M} \leq b - \sqrt{\frac{\log(2/\delta)}{2 M}}  \implies \zx < 1/2 \\
    &{\rm if\quad }\frac{k(\vec x)}{M} \geq b + \sqrt{\frac{\log(2/\delta)}{2 M}}  \implies \zx > 1/2 \\
    &{\rm if \quad}\left\vert\frac{k(\vec x)}{M} - b  \right\vert \leq\sqrt{\frac{\log(2/\delta)}{2 M}} \implies \zx \approx 1/2.
\end{align}
The last condition indicates  that is impossible to distinguish whether $\zx > 1/2$ or $\zx \leq 1/2$. Hence, a quantum classifier with $M$ samples is capable of correctly classify a data point $\vec{x}$, with probability $1 - \delta$ if
\begin{equation}
    \zx \leq b - \sqrt{\frac{\log(2/\delta)}{2 M}}, 
\end{equation}
yielding the desired result. 
\qed
\section{Proof of \Cref{le:chebyshev}}\label{app:chebyshev}

To begin the proof, we first define failure in a classification task, which includes two cases: when the sample is misclassified, and when the sample is close enough to the decision boundary that the measurements used to determine the class do not yield a conclusive result. This allows us to identify
\begin{equation}
    \operatorname{Prob}_F\left( \hat Z_y^{(b)}, \mathcal X\right) \equiv 
    \prob{\zx \geq b - \sqrt{\frac{\log(2 / \delta)}{2M}}}, 
\end{equation}
where the success probability is at least $1 - \delta$, according to~\Cref{le:classification}. The next step is to use Chebyshev's inequality, stated as follows. Let $X$ be a random variable with variance $\sigma^2$. Then
\begin{equation}
    \prob{\vert X - \mathds{E}[X]\vert \geq k} \leq \frac{\sigma^2}{k^2}. 
\end{equation}

We just need to identify terms in Chebyshev's inequality to find
\begin{equation}
    \prob{\zx \geq b - \sqrt{\frac{\log(2 / \delta)}{2M}}} \leq \frac{\sigmax[2]}{\left(b - \mutx[1] - \sqrt{\frac{\log(2/\delta)}{2 M}}\right)^2}, 
\end{equation}
\qed

\section{Proof of~\Cref{le:Bernstein}}\label{app:berstein}
\lebernstein*

We begin by stating the following result, known as one of Berstein's inequalities. 
\begin{theorem}[Berstein's inequality~\cite{Bernstein, bercu2015concentration}]
    Let $X_1,...,X_n$ be zero-mean independent random variables. If, for every $X_i$ for $i\in \{1,...,n\}$ and $t \geq 2$, there exists a positive constant $L$ such that
    \begin{equation}
        \mathbb E\left[ \vert X_i \vert^t\right] \leq \frac{1}{2} \mathbb E\left[X_i^2\right] L^{t - 2} t!, 
    \end{equation}
    then 
    \begin{equation}
        \prob {\sum_{i= 1}^n X_i\geq 2 k \sqrt{\sum_{i= 1}^n \mathds{E}\left[X_i^2\right]}}  < \exp(-\frac{k^2}{2 \left( \sum_{i= 1}^n \mathds{E}\left[X_i^2\right] + L k\right)}) \quad {\rm for }\quad  0 \leq k \leq \frac{1}{2L}\sqrt{\sum_{i= 1}^n \mathds{E}\left[X_i^2\right]}.
    \end{equation}
\end{theorem}
We consider $\zx$ as our only random variable, hence $n = 1$. . Since this variable has a non-zero mean, we use the corresponding standardized moments to apply Bernstein's inequality:
\begin{equation}
    \mathds{E}\left[\left\vert \zx - \mu_1\right\vert^t\right]\leq \frac 12 \sigma^2 L^{t-2}t! \ .
    \label{eq:modified_Brenstein_condition}
\end{equation}
For simplicity of the notation, we have defined $\mu_1:= \mu_1(\hat Z, \mathcal X)$ and $\sigma:= \sigma(\hat Z, \mathcal X)$. We can further relax the condition in Eq.~\eqref{eq:modified_Brenstein_condition} as follows. First, we consider the $t$-th square rooth of Eq.~\eqref{eq:modified_Brenstein_condition}. 
\begin{equation}
    \left(\mathds{E}\left[\left\vert \zx - \mu_1\right\vert^t\right]\right)^{1/t}\leq \frac{1}{2^{1/t}}\sigma^{2/t}L^{1- 2/t}(t!)^{1/t}.
    \label{eq:Br_condition_1/t}
\end{equation}
Now, noting that trivially $\sigma\leq 1$ and $L\leq1$ because of the range of $z(\bm x)\in [0,1]$ and that the following inequality holds for $t\geq 1$ 
\begin{equation}
    (t!)^{1/t}\geq (2 \pi t)^{1/t}\frac{t}{e},
\end{equation}
we can find a lower bound for the right hand side of Eq.~\eqref{eq:Br_condition_1/t}:
\begin{equation}
    \frac{1}{2^{1/t}}\sigma^{2/t}L^{1- 2/t}(t!)^{1/t}\frac t e \geq \sigma^2 L (\pi t)^{1/t} \frac t e \geq \sigma^2 L \pi \frac t e .
    \label{eq:lower_bound_brenstein}
\end{equation}
in the first inequality we have used that $\sigma^{2/t}\geq \sigma^2$ and $L^{1-2/t}\geq L$ considering that $t\geq 1$
and $\sigma,L \leq 1$. We have also used Stirling's approximation to bound
\begin{equation}
    (t!)^{1/t} \geq (2\pi t)^{1/t} \frac{t}{e}, 
\end{equation}
being $e$ Euler's constant. For the second inequality in Eq.~\eqref{eq:lower_bound_brenstein}, we have used that $\pi^{1/t}\geq \pi$ and that $t^{1/t}t\leq t$ considering $t\geq 1$.  Putting all together, we can relax the condition for Berstein's inequality. If the following holds
\begin{equation}
    \left(\mathds{E}\left[\left\vert \zx - \mu_1\right\vert^t\right]\right)^{1/t} \leq \sigma^2 L \pi \frac t e,
\end{equation}
then~\Cref{eq:Br_condition_1/t} applies.

Finally, we need to consider the classification task. For the data point $\vec x$ to be misclassified or not-determined, we need that
\begin{equation}
    \zx \geq b - \sqrt{\frac{\log(2 / \delta)}{2 M}}. 
\end{equation}
Therefore, the probability of failure in the classification is given by
\begin{equation}
    \prob{\zx \geq b - \sqrt{\frac{\log(2 / \delta)}{2 M}}} \leq \exp\left( - \frac{k^2}{2\left(\sigma^2 + L  k \right) }\right), 
\end{equation}
where we have taken $k =\left[ b - \sqrt{\frac{\log(2 / \delta)}{2 M}} - \mu_1\right]$, yielding the desired result. \qed

% which has to fulfill the condition in Eq.~\eqref{eq.berstein2}, which implies 
% \begin{equation}
%     \frac{1}{2} -  \sqrt{\frac{\log(2 / \delta)}{2 M}} -\frac{\sigma^2}{L} \leq  \mu_1 \leq \frac{1}{2}- \sqrt{\frac{\log(2 / \delta)}{2 M}}
% \end{equation}
% \bcf{I don't understand the intuition of this condition}
% Additionally, for $\mu_1$ smaller than the value above, 
% \begin{equation}
%     \prob{\zx \geq \frac 12 - \sqrt{\frac{\log(2 / \delta)}{2 M}}} \leq \exp\left( -\frac{\sigma^2}{\aps{4} L^2}\right).
% \end{equation}

% \aps{We need to write this better, and make results compact, but the maths are done.}
% \qed

\section{Proof of~\Cref{le:subgaussian}}\label{Appendix:sub-gaussian}
\lesubgaussian*
We begin by defining the sub-gaussianity condition. 
\begin{definition}[Sub-gaussianity]
A zero-mean random variable $X$ is sub-gaussian if there exists $C > 0$ such that 
\begin{equation}
    \prob{\vert X \vert \geq k} \leq 2 \exp\left( \frac{-k^2}{C^2}\right).
    \label{eq:definition_subgaussian}
\end{equation}
\end{definition}
The condition 
\begin{equation}\label{eq.gamma_subgaussian}
    \mathds E\left[ \vert X \vert^t\right] \leq 2 L^{t} \Gamma\left(  \frac{t}{2} + 1 \right),
\end{equation}
for $p$ a positive constant, is equivalent to the condition in~\Cref{eq:definition_subgaussian}, as we will show. By Markov's inequality, for all sub-gaussian variables $X$
\begin{equation}\label{eq.markov}
    \prob{\vert X \vert \geq k}= \prob{\exp\left( \frac{X^2}{C^2}\right) \geq \exp\left( \frac{k^2}{C^2}\right)} \leq \mathds E\left[ \exp\left( \frac{X^2}{C^2}\right) \right] \exp\left( \frac{-k^2}{C^2}\right) \leq 2 \exp\left( \frac{-k^2}{C^2}\right).
\end{equation}
Therefore, we just need to show
\begin{equation}
    \mathds E\left[ \vert X \vert^t\right] \leq 2 L^{p} \Gamma\left(  \frac{t}{2} + 1 \right) \implies
    \mathds E\left[ \exp\left( \frac{X^2}{C^2}\right) \right] \leq 2.
\end{equation}
To do so, we expand by Taylor
\begin{equation}
    \mathds E\left[ \exp\left( \frac{X^2}{C^2}\right) \right] = 1 + \sum_{t = 0}^\infty \frac{\mathds E\left[ X^{2t} \right]}{C^{2t} t!} \leq  1 + \sum_{t = 0}^\infty \frac{2 L^{2t} \Gamma\left(t + 1\right)}{C^{2t} t!} = 1 + 2 \sum_{t = 1}^\infty \left( \frac{L^2}{C^2}\right)^t. 
\end{equation}
The last sum can be identified as the Taylor expansion of $f(x) = (1 - x)^{-1}$, thus
    \begin{equation}
        \mathds E\left[ \exp\left( \frac{X^2}{C^2}\right) \right] \leq \frac{2}{1 - \frac{L^2}{C^2}} - 1.
    \end{equation}
Connecting the previous result to~\Cref{eq.markov} we just need to impose
\begin{equation}
    \frac{2}{1 - \frac{L^2}{C^2}} - 1 \leq 2 \Rightarrow C \geq \sqrt 3 L.
\end{equation}
Hence 
\begin{equation}
    \mathds{E}\left[ \vert X \vert^t\right] \leq 2 L^{p} \Gamma\left(  \frac{t}{2} + 1 \right) \Rightarrow \prob{\vert X \vert \geq k} \leq 2 \exp\left( \frac{-k^2}{3 L^2}\right).
\end{equation}

Now, following the steps of~Appendix \ref{app:berstein}, we recall
\begin{equation}
    2^{1/t} \geq 1
\end{equation}
and 
\begin{equation}
    \Gamma\left( \frac{t}{2} + 1\right)^{1/t} \geq (\pi t)^{1/2t} \sqrt{\frac{t}{2e}} \geq \frac{t}{2e}.
\end{equation}
We can thus relax the condition in~\Cref{eq.gamma_subgaussian} to state that 
\begin{equation}
    \mutx^{1/t} \leq \frac{L}{\sqrt{2 e}} \sqrt t,
\end{equation}
which implies
\begin{equation}
    \prob{\zx \geq b - \sqrt{\frac{\log(2 / \delta)}{2 M}}} \leq \exp\left( - \frac{k^2}{3 L^2}\right), 
\end{equation}
where we have taken $k =\left[ b - \sqrt{\frac{\log(2 / \delta)}{2 M}} - \mu_1\right]$, yielding the desired result. \qed

\section{Proof of \Cref{th.antirandomness_dlp}}\label{app:antirandomness_dlp}
\thmantirandomness*

\begin{proof}
We just need to compute the average anti-randomness for $t = 1$ and $ t= 2$. 

We first make use of the DLP classification problem. In this problem, we define two hyperplanes that exist for every concept class $y_s\in \mathcal{C}$: 
\begin{equation}
        \begin{split}
            |\psi_s^{(1)}\rangle  &= \frac{1}{\sqrt{(p-1)/2}}\sum_{i = 0}^{(p-3)/2}\ket{g^{s+i}}\\
            |\psi_s^{(0)}\rangle  &= \frac{1}{\sqrt{(p-1)/2}}\sum_{i = (p-1)/2}^{p-1}\ket{g^{s+i}},
        \end{split}
\end{equation}
        which define two projectors $\Pi_0 = |\psi_s^{(0)}\rangle \langle\psi_s^{(0)}| $ and $\Pi_1 = |\psi_s^{(1)}\rangle \langle\psi_s^{(1)}|$ with the following properties~\cite{liu2021rigorous}:
        \begin{itemize}
            \item $\langle \psi(x)|\Pi_1| \psi(x) \rangle = \Delta$, for a fraction $1-\Delta$ of $x$ such that $y(x)=1$.
            \item $\langle \psi(x)|\Pi_1| \psi(x) \rangle = 0$, for a fraction $1-\Delta$ of $x$ such that $y(x)=0$.
            \item $\langle \psi(x)|\Pi_1| \psi(x) \rangle \leq \Delta$, for a fraction $\Delta$ of $x$ such that $y(x)=1$.
            \item $\langle \psi(x)|\Pi_0| \psi(x) \rangle \leq \Delta$, for a fraction $\Delta$ of $x$ such that $y(x)=1$.
        \end{itemize}
        We have used the quantity  
        \begin{equation}
            \Delta =\frac{2^{k+1}}{p} \in \Theta(1/\text{poly}(n)), \quad \text{with } k = n - c\log n,
        \end{equation}
        being $c$ a constant. We are interested in the observable $\hat Z_ s$, which in this scenario is defined by 
\begin{equation}
    \Zdlp = \frac{\mathds{I}+( \Pi_0-\Pi_1)(-1)^{y_s(x)}}{2}.
    \label{eq:Z_s_DLP_app}
\end{equation}
For a given $|\psi(x)\rangle \in \mathcal{X}_g$, the expectation value of $\hat Z_s$ in this state will give a smaller value than $1/2$ if the classification is correct and higher than $1/2$ if the classification is incorrect. Notice that the symmetry of this problem allows us to treat both classes analogously. 
        
We can now bound $\mu_1(\hat Z_s, \mathcal X_g)$ as 
\begin{align}
\label{eq:upper_bound_mu1_Zs}
    \mu_1(\hat Z_s, \mathcal X_g) & \leq (1 - \Delta) \frac{1 - \Delta}{2} + \Delta \frac{1 + \Delta}{2} = \frac{1 - \Delta}{2} + \Delta^2 \\ 
    \mu_1(\hat Z_s, \mathcal X_g) & \geq \frac{1 - \Delta}{2}.
\label{eq:lower_bound_mu1_Zs}
\end{align}

On the other hand, we can compute $\mu_1(\Zdlp)$ over Haar-random states making use of~\Cref{eq:first_moment_Haar} and considering the symmetry in the eigenspace of $\Zdlp$: 
\begin{equation}
    \mu_1(\Zdlp) = \frac{1}{2^{n-1}}\left(\frac 12 + \frac{2^n-2}{4}\right) = \frac 12,
\end{equation}
where we take into account that the eigenvalues of $\Zdlp$ are $\vec \lambda = (1, 0, 1/2)$ with multiplicities $\vec m = (1,1,2^n -2)$. With this, we can express the anti-randomness for $t= 1$ as
\begin{equation}
    \mathcal A_1^{(\Zdlp)}(\mathcal{X}_g) = \left\vert\frac{1}{2}-\mu_1(\hat Z_s, \mathcal X_g)\right\vert,
\end{equation}
which we can bound as 
\begin{equation}
     \frac{\Delta}{2} - \Delta^2\leq  \mathcal A_1^{(\Zdlp)}(\mathcal{X}_g) \leq \frac{\Delta}{2} ,
\end{equation}
which ensures us that the 1-antirandomness scales as $\Theta(1/{\rm poly}(n))$, thus our set of states is polynomially bounded-away from the first $\Zdlp-$shadowed Haar-random moment. 

The second moment can be upper bounded as
\begin{equation}
    \mu_2(\hat Z_s, \mathcal X_g) \leq (1 - \Delta) \left(\frac{1 - \Delta}{2} \right)^2 + \Delta \left(\frac{1 + \Delta}{2} \right)^2 = \left(\frac{1 - \Delta}{2} \right)^2 + \Delta^2,  \label{eq:bound_var_DLP}
\end{equation}
hence 
\begin{equation}
    \sigma^2(\Zdlp, \mathcal{X}_g) =  \mu_2(\hat Z_s, \mathcal X_g)-  \mu_1(\hat Z_s, \mathcal X_g)^2\leq \Delta^2.
\end{equation}
Now, for the average anti-randomness for $t=2$ we compute the second standardized moment $\mutbarhaar[2]{\Zdlp}$. In~\Cref{eq:variance_Haar_states}, we have derived an expression for the variance, which takes the following simple form when considering the eigenbasis of $\Zdlp$: 
\begin{equation}
    \sigma^2(\Zdlp) = \frac{1}{2^{n-1}+1}.
\end{equation}
Now, for the $t = 2$ average anti-randomness: 
\begin{equation}
    \mathcal A_2^{(\Zdlp)}(\mathcal{X}_g) = \left\vert  \frac{1}{2^{n-1}+1}- \sigma^2(\Zdlp, \mathcal X_g)\right\vert \in \Theta\left( \operatorname{poly}^{-1}\right).
\end{equation}
This result finishes the proof. 
%Both results show that the average randomness is bounded away from that of Haar-random states by at least a polynomial margin, finishing the proof. 
\end{proof}

\section{Proof of \Cref{le:DLP1}}
\label{appendix:variance_DLP}
\lemDLP*
\begin{proof}
For this proof, we just need to use the bounds in\Cref{eq:lower_bound_mu1_Zs} and \eqref{eq:bound_var_DLP}. These bounds applied to~\Cref{le:chebyshev} allow us to bound the probability of failing in the classification as
\begin{equation}
    \operatorname{Prob}_F \left(  \Zdlp, \mathcal X_{g}\right)  \leq \frac{\Delta^2}{\left( \frac{\Delta}{2} - \Delta^2 - \sqrt{\frac{\log(2 / \delta)}{2 M}}\right)^2} = \frac{1}{\left( \frac{1}{2} - \Delta - \sqrt{\frac{\log(2 / \delta)}{2 M \Delta^2}}\right)^2}.
\end{equation}
Choosing the number of measurements $M = \log (2/\delta) / 2 \left(\Delta / 2 - \Delta^2 + \Delta^{1 / 2} \right)^{-2}\in \Theta(\operatorname{poly}(n))$ we can obtain 
\begin{equation}
    \operatorname{Prob}_F \in \mathcal O\left(\operatorname{poly}^{-1}(n)\right)
\end{equation}
\end{proof}

\section{Proof of \Cref{th:counterexample}}\label{app:counterexample}
\thcounterexample*
\begin{proof}
We have the following expectation value: 
\begin{equation}
    z(\vec x) = \frac{1}{2} - \sqrt{\vec x_{\lfloor n/2\rfloor } \vec x_{\lceil n/2\rceil } },
\end{equation}
and we want to compute its $t$-moments, this is $\mathds E [z(\vec x)^t]$. Let's go step by step: 
\begin{align}
    \mathds E [z(\vec x)^t] = &\mathds E \left[\left(\frac{1}{2} - \sqrt{\vec x_{\lfloor n/2\rfloor } \vec x_{\lceil n/2\rceil } }\right)^t \right] = \mathds E \left[\sum_{k=0}^t (-1)^k \binom{t}{k} \left(\frac{1}{2}\right)^{t-k} \left(\vec x_{\lfloor n/2\rfloor} \vec x_{\lceil n/2\rceil }\right) ^{k/2}\right] =\\& \sum_{k=0}^t (-1)^k \binom{t}{k} \left(\frac{1}{2}\right)^{t-k} \mathds E \left[\left(\vec x_{\lfloor n/2\rfloor} \vec x_{\lceil n/2\rceil }\right) ^{k/2}\right].
    \label{eq:t_exp_value_Dirichlet}
\end{align}
Recalling that the vector $\vec x$ follows a Dirichlet distribution with parameter $\alpha_i = \frac 12 \binom n i $ and $i\in \{0,1,..,n\}$, we apply the the following expression for the $t$-moments of a Dirichlet distribution: 
\begin{equation}
    \mathds E \left[\prod_{i = 0}^k \vec x_i ^{\beta_i}\right] = \frac{\Gamma\left(\sum_{i = 0}^k \alpha_i\right)}{\Gamma \left[ \sum_{i = 0}^k \left(\alpha_i + \beta_i \right)\right]} \prod_{i = 0}^k\frac{\Gamma\left(\alpha_i + \beta_i \right)}{\Gamma(\alpha_i)}.
\end{equation}
Comparing this expression with the last term in~\Cref{eq:t_exp_value_Dirichlet}, we identify that $\beta_i = 0$ for all $i$ except $\beta_{\lfloor n/2\rfloor} = \beta_{\lceil n/2\rceil} = k/2$. Also taking into consideration that $\alpha_i = \frac 12 \binom n i $, we can express the $t$-moment as 
\begin{align}
    \mathds E [z(\vec x)^t] ) = &\sum_{k=0}^t (-1)^k \binom{t}{k} \left(\frac{1}{2}\right)^{t-k} \frac{\Gamma(2^{n - 1})}{\Gamma(2^{n - 1} + k)} \frac{\Gamma\left( \frac 12 \binom{n}{\lceil n/2\rceil} + k/2\right)}{\Gamma\left(\frac 12 \binom{n}{\lceil n/2\rceil}\right)}\frac{\Gamma\left( \frac 12 \binom{n}{\lfloor n/2\rfloor} + k/2\right)}{\Gamma\left(\frac 12 \binom{n}{\lfloor n/2\rfloor}\right)} =\\ &\sum_{k=0}^t (-1)^k \binom{t}{k} \left(\frac{1}{2}\right)^{t-k} \frac{\Gamma(2^{n - 1})}{\Gamma(2^{n - 1} + k)} \left(\frac{\Gamma\left( \frac 12 \binom{n}{\lceil n/2\rceil} + k/2\right)}{\Gamma\left(\frac 12 \binom{n}{\lceil n/2\rceil}\right)}\right)^2,
\end{align}
where we have considered that when $n$ is odd (our initial assumption), then $\binom{n}{\lceil n/2\rceil} = \binom{n}{\lfloor n/2\rfloor}$. 

Now, we compute mean and variance.
\begin{equation}
\label{eq:counterexample_moment1_app}
     \mathds E [z(\vec x)] = \frac 12 -  \frac{\Gamma(2^{n - 1})}{\Gamma(2^{n - 1} + 1)} \left(\frac{\Gamma\left( \frac 12 \binom{n}{\lfloor n/2\rfloor} + 1/2\right)}{\Gamma\left(\frac 12 \binom{n}{\lfloor n/2 \rfloor }\right)}\right)^2 = \frac 12 -\frac{1}{2^{n-1}} \left(\frac{\Gamma\left( \frac 12 \binom{n}{\lfloor n/2\rfloor} + 1/2\right)}{\Gamma\left(\frac 12 \binom{n}{\lfloor n/2 \rfloor }\right)}\right)^2, 
\end{equation}
where we have used that $\Gamma(x+1) = x \Gamma(x)$. For the second moment, we have 
\begin{align}
     \mathds E [z(\vec x)^2] = &\mathds E\left[ \left(\frac{1}{2} - \sqrt{\vec x_{\lfloor n/2\rfloor } \vec x_{\lceil n/2\rceil } }\right)^2\right] = \frac{1}{4}-  \mathds{E}\left[ \sqrt{\vec x_{\lfloor n/2\rfloor }\vec x_{\lceil n/2\rceil }}\right] + \mathds{E}\left[ \vec x_{\lfloor n/2\rfloor }\vec x_{\lceil n/2\rceil }\right]  = \\&
     \frac{1}{4} - \frac{\Gamma(2^{n - 1})}{\Gamma(2^{n - 1} + 1)} \left(\frac{\Gamma\left( \frac 12 \binom{n}{\lceil n/2\rceil} + 1/2\right)}{\Gamma\left(\frac 12 \binom{n}{\lceil n/2\rceil}\right)}\right)^2 + \frac{\Gamma(2^{n - 1})}{\Gamma(2^{n - 1} + 2)} \left(\frac{\Gamma\left( \frac 12 \binom{n}{\lceil n/2\rceil} + 1\right)}{\Gamma\left(\frac 12 \binom{n}{\lceil n/2\rceil}\right)}\right)^2 = \\& 
     \frac 14 - \frac{1}{2^{n-1}}\left(\frac{\Gamma\left( \frac 12 \binom{n}{\lceil n/2\rceil} + 1/2\right)}{\Gamma\left(\frac 12 \binom{n}{\lceil n/2\rceil}\right)}\right)^2 + \frac{\frac{1}{2^2}\binom{n}{\lceil n/2\rceil}^2}{(2^{n-1}+1 )2^{n-1}}.
\end{align}
With this, we can compute the variance:
\begin{align}
    \operatorname{Var} [z(\vec x)] = \mathds E [z(\vec x)^2]- \mathds E [z(\vec x)]^2 = \frac{\frac{1}{2^2}\binom{n}{\lceil n/2\rceil}^2}{(2^{n-1}+1 )2^{n-1}}- \left(\frac{1}{2^{n-1}}\right)^2 \left(\frac{\Gamma\left( \frac 12 \binom{n}{\lfloor n/2\rfloor} + 1/2\right)}{\Gamma\left(\frac 12 \binom{n}{\lfloor n/2 \rfloor }\right)}\right)^4.
\end{align}
Next, we want to bound $\mu_1$ to determine its scaling. In particular, we want to obtain a lower bound for the margin (how much it deviates from $1/2$). We are going to use the following inequality, often refereed as Gautschi's inequality~\cite{Wendel, Gautschi}:
\begin{equation}
    1  \geq \frac{\Gamma(x+s)}{\Gamma(x)x^s} \geq \left(\frac{x}{x+s}\right)^{1-s},
\end{equation}
for $x>0$ and for $s\in(0,1)$. In particular, we are going to use it for $s = 1/2$. Applying the inequality to~\Cref{eq:counterexample_moment1_app} followed by the Stirling approximation $\binom{n}{\lfloor n/2\rfloor} \approx \frac{2^n}{\sqrt{\pi n/2}}$, we have 
\begin{equation}
    \mathds E [z(\vec x)] \leq \frac 12 -  \frac{1}{2^{n-1}}\frac{\frac{1}{4}\binom{n}{\lceil n/2 \rceil }^2}{\frac{1}{2}\binom{n}{\lceil n/2 \rceil }+ \frac{1}{2}}\approx \frac 12 - \frac{2\sqrt 2}{\sqrt{\pi n}}.
\end{equation}
Therefore, we can express 
\begin{equation}
    \frac 12 -  \mathds E [z(\vec x)] \geq \frac{\sqrt{2}}{\sqrt{\pi n}}   .
\end{equation}
For the variance, we can use again Gautschi's inequality together with the triangular inequality, to bound it as follows
\begin{equation}
    \operatorname{Var} [z(\vec x)] \leq \frac{1}{2^2}\binom{n}{\lceil n/2\rceil}^2 \left( \frac{1}{2^{2(n - 1)}} - \frac{1}{2^{n - 1} (2^{n - 1} + 1)}\right) = \binom{n}{\lceil n/2\rceil}^2 \left( \frac{1}{2^{2n} (2^{n - 1} + 1)}\right).
\end{equation}
Using Stirling again, we find
\begin{equation}
    \operatorname{Var} [z(\vec x)] \leq \frac{2}{\pi (2^{n-1} + 1)}.
\end{equation}
Making use of \Cref{le:chebyshev} we can bound the probability of misclassification as
\begin{equation}
    \operatorname{Prob}_F \left( \hat O_X, \mathcal X_{\mathcal{B}, \mathcal{D}}\right) \leq \frac{2^{-n}}{2 \pi} \left( \sqrt{\frac{8}{\pi n}} - \sqrt{\frac{\log(2 / \delta)}{2 M}}\right)^{-2}, 
\end{equation}
which implies that the probability of failure scales as
\begin{equation}
    \operatorname{Prob}_F \left( \hat O_X, \mathcal X_{\mathcal{B}, \mathcal{D}}\right)\in \exp\left(-\Omega(n)\right), 
\end{equation}
for $M \in \Omega(n)$.
\end{proof}
For completeness, we provide approximations in the case where $t \ll 2^n$. We can use Stirling's approximation $\Gamma(x + a) \approx \Gamma(x) x^a$, for $a \ll x$: 
\begin{align}
    \mathds E [z(\vec x)^t] ) \approx &\sum_{k=0}^t (-1)^k \binom{t}{k} \left(\frac{1}{2}\right)^{t-k} \left(\frac{1}{2^{n-1}}\right)^k \left[\frac 12\binom{n}{\lfloor n/2\rfloor}\right]^{k} =
    \left(\frac 12- \frac{1}{2^n}\binom{n}{\lfloor n/2\rfloor}\right)^t.
\end{align}
If we assume that $n$ is large, then we can approximate $\binom{n}{\lfloor n/2\rfloor} \approx \frac{2^n}{\sqrt{\pi n/2}}$: 
\begin{align}
    \mathds E [z(\vec x)^t] \approx \left(\frac 12- \frac{1}{\sqrt{\pi n/2}}\right)^t = \left(\frac{1}{2}\right)^t \left(1- \frac{2\sqrt{2}}{\sqrt{\pi n}}\right)^t, 
\end{align}
implying that the random variable $z(\vec x)$ concentrates sufficiently to apply \Cref{le:Bernstein} and \Cref{le:subgaussian}.

\section{Details on the experiments}
\label{Appendix:learning_problem}
In this section, we introduce the two QML models used as examples to analyze their data-induced randomness. Additionally, we describe the two-dimensional classification learning problem selected for this study.

The first model we employ is a feature-map classifier inspired by Ref.~\cite{havlicek2019supervised}. The classifier consists of two main components: a fixed feature map $W(\vec x)$ and a variational circuit $U(\vec\theta)$. Hence, our set of states is
\begin{equation}
    \xf = \biggl\{ \ket{\psi_{\vec\theta}(\vec x)} \equiv U(\vec\theta)W(\vec x) \ket 0 \biggl\}_x.
\end{equation}
The particular choice of $W(\vec x)$ and $U(\vec\theta)$ is inspired by Ref.~\cite{havlicek2019supervised}. We extend the feature-maps proposed to multi-qubit scenarios as described in~\Cref{fig:circuits} $(a)$ and $(b)$. The variational circuit that we apply after the feature map is a hardware efficient ansatz~\cite{cerezo2021cost},
\begin{equation}
    W(\vec{\theta}) = \prod_{m = 1}^L \left( \bigotimes _{k = 1}^nR_y(\theta_{mk}) \prod_{i,j\in \mathcal{I}} \operatorname{CNOT}_{i,j} \right),
\end{equation}
where $\operatorname{CNOT}_{i,j}$ is the controlled not gate between qubits $i$ and $j$, and $\mathcal{I}$ is a set of indices. In our case, $\mathcal{I}$ is the set of indices with first neighbour connectivity and periodic boundary conditions. 

The second model under consideration is the data re-uploading~\cite{perez-salinas2020data}, in which the encoding and the training process are interleaved in the quantum circuit. The set of states is now 
\begin{equation}\label{eq:reuploading}
    \xru = \left\{ \ket{\psi_{\vec\theta}(\vec x)} \equiv \prod_{l = 1}^L U(\vec\theta_l, x) \ket 0 \right\}_x.
\end{equation}
where $\vec\theta_l$ are the trainable parameters. 
In our experiments, the encoding and trainable gates are interleaved in a layer-wise structure, see~\Cref{eq:reuploading}. In particular, we use an ansatz given by~\Cref{fig:circuits} $(c)$. This block constitutes a layer, and we repeat it $L$ times.

\begin{figure}
    \centering
    \includegraphics[width=1\linewidth]{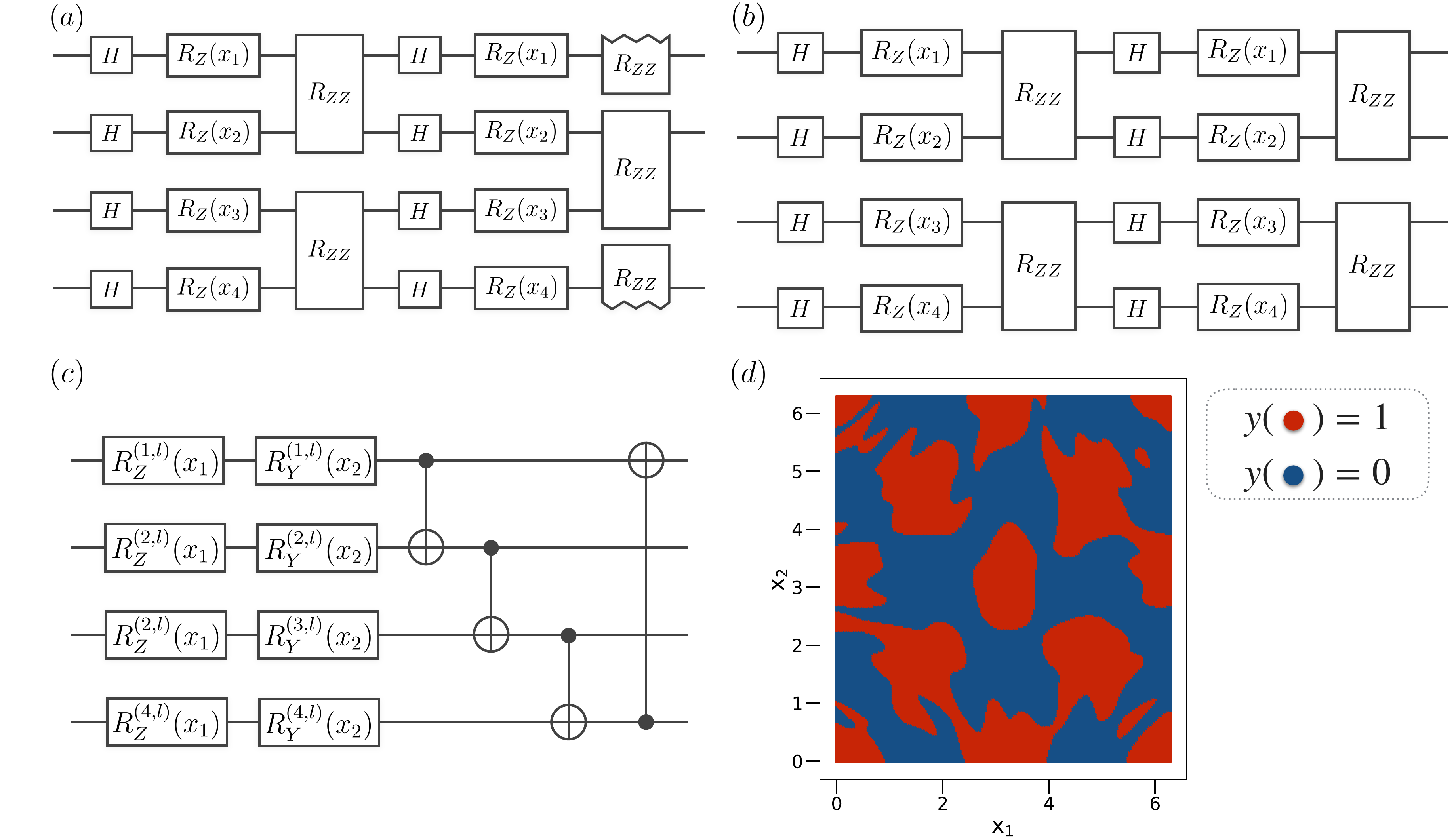}
    \caption{Quantum circuits for encoding the dataset $\vec{x}$ and classification pattern. The gate $R_{ZZ}(x_i, x_j)$ is defined as $R_{ZZ}\left(2(\pi - x_i)(\pi - x_j)\right)$, where $i$ and $j$ denote the $i$-th and $j$-th qubits, respectively. $({a}) $ Brick-layer encoding. $({b})$ Non-brick-layer encoding. Circuit that represents a single layer $l$ of the data re-uploading model used in this work. We have defined the rotation gates $R_Z(n,l)(x_1) = R_Z(\theta^{n,l}_1 x_1 + \theta^{n,l}_2)$ and $R_Y(n,l)(x_1) = R_Y(\theta^{n,l}_3 x_2 + \theta^{n,l}_4)$. $(c)$ A single layer of the data re-uploading circuit utilized in our study. $(d)$ The classification pattern that the model aims to learn.
 }
    \label{fig:circuits}
\end{figure}

The learning problem that we use in both the feature map and data re-uploading classifiers is a two-dimensional classification problem. We have a training set given by $\{\mathcal{X}, \mathcal{Y}\}$, where $\mathcal{X}$ is the two-dimensional data set  $(x_1, x_2)$ and $\mathcal{Y}$ their associated labels, which can take values $y\in \{1,0\}$. The goal of the learning algorithms is to find a function $g: \mathcal{X}\rightarrow \mathcal{Y}$ that correctly labels most of the input data. In our case, this function is the sign of the expectation value. In particular, the observable that we use is $\sigma^{(z)} := \sigma_1^{(z)}\otimes \sigma_2^{(z)}\otimes...\otimes \sigma_n^{(z)}$, being $\sigma_i^{(z)}$ the $Z-$Pauli matrix acting on the $i-$th qubit. 

The observable that we use in the loss function is  $\hat{Z}_y = \frac{1}{2}(\mathds{I}-y(\vec x)\sigma^{(z)})$, where $y(\vec x)$ is the correct label associated with $\vec x$. Therefore, the loss function is given by 
\begin{equation}
      \mathcal{L}(\vec{\theta}) = \sum_{\vec{x}\in X_{\text{train}}} \langle \psi(\vec{x}, \vec{\theta})|\hat{Z}_y|\psi(\vec{x}, \vec{\theta} )\rangle.
\end{equation}
The predicted label is given by  $y' = \text{sign}\left( \langle \psi(\vec{x}, \vec{\theta} )|\sigma^{(z)}|\psi(\vec{x}, \vec{\theta} )\rangle \right)$. For data points with a true label $y(\vec x)= 1$ ($y(\vec x)) = 0)$, the loss function rewards the expectation value of $\sigma^{(z)}$ being positive (negative). Ideally, the algorithm learns the separating hyperplane in the feature-map model and the optimal data mapping in the data re-uploading model to achieve accurate classification. We employ the gradient descent-based L-BFGS-B optimization algorithm to minimize the loss function. 

In both models, the random variable used for signaling correctly or wrongly classified data points is 
\begin{equation}
   \zx = \frac 12 \langle \psi_{\vec\theta}(\vec x) \vert \left( \mathds{I} - y(\bm x) \sigma^{(z)}\right) \vert \psi_{\vec\theta}(\vec x)\rangle,
\end{equation}
where $y(\vec x)$ is the true data-label associated to $\vec x$. 

Finally, we discuss how we design the classification pattern that we want the models to learn. We could have used a regular pattern, like points inside or outside of the circuit. Instead, we chose to work with the pattern that is proposed in Ref. \cite{havlicek2019supervised}. The dataset is created synthetically according to the following quantity: if $\langle E(\vec{x_i})| V^\dagger \sigma^{(z)} V |E(\vec{x_i})\rangle>0$, then $y(\vec x_ i)= 1$, and $y(\vec x_i)=0$ otherwise. We have defined the encoding vector $|E(\vec{x_i})\rangle$ as the resulting state of applying the  feature map (Fig. \ref{fig:circuits} $(a)$) to the initial state. We choose $V$ to be a random unitary matrix sampled from $SU(2^n)$. 
For simplicity, we generate a single dataset for $n= 2$ and use it consistently across all models. The classification pattern is illustrated in~\Cref{fig:circuits} $(d)$. For every choice of the unitary $V$, we create a different dataset.

\end{document}